\documentclass[runningheads]{llncs}
\usepackage[T1]{fontenc}
\usepackage{graphicx}
\usepackage{color}
\usepackage{enumitem}
\usepackage{booktabs}
\usepackage{bigstrut}
\usepackage{siunitx}
\usepackage{amssymb}
\usepackage{textcomp}
\usepackage{mathtools}
\usepackage[plain]{algorithm}
\usepackage{algpseudocode}
\usepackage{fp}
\usepackage{tikz}
\usepackage{wrapfig}

	

\makeatletter
\let\OldStatex\Statex
\renewcommand{\Statex}[1][3]{%
	\setlength\@tempdima{\algorithmicindent}%
	\OldStatex\hskip\dimexpr#1\@tempdima\relax}
\makeatother

\let\oldReturn\Return
\renewcommand{\Return}{\State\oldReturn}

\usetikzlibrary{decorations.markings}

\tikzset{
	arrow/.style={decoration={markings, mark=at position 1 with
			{\arrow[scale=1.5,>=stealth]{>}}}, postaction={decorate}},
	arrow/.default=1
}

\newcommand*{\xscale}{0.093}
\newcommand*{\yscaleratio}{0.2}
\newcommand*{\yscaletime}{0.05}
\newcommand*{\yscaleRatioLarge}{\yscaleratio}
\newcommand*{\yscaleTimeLarge}{\yscaletime}
\FPeval\yscaleRatioMiddle{1.666667*\yscaleratio}
\FPeval\yscaleTimeMiddle{1.666667*\yscaletime}
\FPeval\yscaleRatioSmall{5*\yscaleratio}
\FPeval\yscaleTimeSmall{5*\yscaletime}
\newcommand*{\axisAdditionalLengthPlus}{0.5}
\newcommand*{\axisAdditionalLengthMinus}{0.15}
\newcommand*{\axisLabel}{0.1}
\newcommand*{\crossSize}{0.075}

\FPeval\xscaleFigureOptimalObjectiveFunctionValuesVsLowerBounds{2*\xscale}
\FPeval\yscaleFigureOptimalObjectiveFunctionValuesVsLowerBounds{15*\yscaleratio}

\FPeval\xscaleEfficiencyStartingHeuristics{4*\xscale}
\FPeval\xscaleEfficiencyImprovementHeuristics{10*\xscale}
\FPeval\yscaleRatioLargeEfficiency{1.67*\yscaleRatioLarge}
\FPeval\yscaleRatioMiddleEfficiency{1.67*\yscaleRatioMiddle}
\FPeval\yscaleRatioSmallEfficiency{1.67*\yscaleRatioSmall}

\newcommand*{\xAxis}{$n$}
\newcommand*{\yAxisRatio}{$\rg$}
\newcommand*{\yAxisTime}{$\ta$}

\newcommand{\rz}{{\mathbb{R}}}

\newcommand*{\defeq}{\mathrel{\vcenter{\baselineskip0.5ex \lineskiplimit0pt
			\hbox{\scriptsize.}\hbox{\scriptsize.}}}%
	=}

\newcommand*{\ie}{i.e.}
\newcommand*{\eg}{e.g.}
\DeclarePairedDelimiterX{\norm}[1]{\lVert}{\rVert}{#1}
\DeclarePairedDelimiterX{\bigNorm}[1]{\big\lVert}{\big\rVert}{#1}
\newcommand*{\la}{{\langle}}
\newcommand*{\ra}{{\rangle}}
\newcommand*{\rg}{\overline{r^*}\genfrac{}{}{0pt}{}{g}{}}
\newcommand*{\ta}{\overline{t^*}^{^{\genfrac{}{}{0pt}{}{}{\scalebox{0.8}{$a$}}}}}

\begin{document}
\title{A new \boldmath$O(n\log n)$\unboldmath\ approach for the Euclidean maximum weight matching problem}
\titlerunning{A new $O(n\log n)$ approach for the Euclidean MWM}
\author{Rostislav Stan\v{e}k\orcidID{0000-0002-0849-9550} \and
Robert Arustamyan\orcidID{0009-0001-2141-3344}}
\authorrunning{R.\ Stan\v{e}k and R.\ Arustamyan}
%\institute{Princeton University, Princeton NJ 08544, USA \and
%Springer Heidelberg, Tiergartenstr. 17, 69121 Heidelberg, Germany
\institute{
	Department of Mathematics and Information Technology, \\ Technical University of Leoben, Peter-Tunner-Stra{\ss}e 25/I, 8700 Leoben, Austria \email{\{rostislav.stanek@unileoben.ac.at, robertarustamyan2@gmail.com\}}
%	\and
%	Department of Mathematics and Information Technology, Technical University of Leoben, Peter-Tunner-Stra{\ss}e 25/I, 8700 Leoben, Austria \email{robertarustamyan2@gmail.com}
}
\maketitle              % typeset the header of the contribution
\begin{abstract}
	In a weighted graph $G = (V, E)$, the maximum weight matching problem (MWM) asks for a matching (\ie\ pairing) of its vertices, such that each vertex is paired with at most one other vertex and the total sum of weights of all edges connecting paired vertices is maximised. If the vertices of the graph correspond to points in the Euclidean plane and the weights to their pairwise Euclidean distances, we get the Euclidean maximum weight matching problem (Euclidean MWM). The best optimum-solution algorithm for this problem runs in $O(n^{2.5})$. Furthermore, there exists an FPTAS guaranteeing a $(1 - \epsilon)$-approximation ratio, which runs in $O(m \epsilon^{-1} \log \epsilon^{-1})$ time. Heuristics with a subquadratic running time (with respect to the number of vertices $|V|$) are known, but often yield solutions of a modest quality.

	In this paper, we present a novel algorithm for solving the Euclidean MWM running in $O(n \log n)$ time and providing excellent solution quality, especially for larger instances. In particular, in our computational tests, the algorithm yielded optimum or near-optimum solutions for all test instances; the worst observed optimality gap was less than $1.38\%$. This makes the algorithm highly attractive for practical applications, especially when exact methods become computationally prohibitive due to the size of the instance.

	\keywords{Euclidean maximum weight matching problem \and Maximum weight matching problem \and Matching \and Heuristic.}
\end{abstract}
\section{Introduction}
	\label{section:introduction}
%	Let $G = (V, E)$ be an undirected graph with vertex set	$V = \{1, 2, \ldots, n\}$ and edge set $E = \big\{\{i, j\} : i, j \in V, i \neq j\big\}$, $|V| = n$, $|E| = m$.
	Let $G = (V, E)$ be an undirected graph with $n$ vertices and $m$ edges; $V = \{v_1, v_2, \ldots, v_n\}$,
%$E = \big\{e_1 = \{e_{11}, e_{12}\}, e_2 = \{e_{21}, e_{22}\}, \ldots, e_m = \{e_{m1}, e_{m2}\}\big\}$.
	$E = \{e_1, e_2, \ldots, e_n\}$, where $e_i = \{e_{i1}, e_{i2}\}$ for each $1 \leq i \leq m.$ A \textbf{matching} is a set of pairwise disjoint edges $M \subseteq E$, where $e_1 \cap e_2 = \emptyset$ for each $e_1 \neq e_2 \in M$. %If there is no danger of confusion, we simply write $i j$ instead of $\{i, j\}$.
	
	Given an undirected graph $G = (V, E)$ and weights $c_e$ for each edge $e \in E$, the \textbf{maximum weight matching problem} (\textbf{MWM}) asks for a matching $M \subseteq E$ with the maximum total weight $f(G, M) \defeq \sum_{m \in M} c_m$; $f$ is called the \textbf{objective function}. The MWM is a well-studied classical combinatorial optimisation problem with direct practical applications in fields such as the social sciences (see, \eg, \textsc{Naini} et al.~\cite{NainiUnnikrishnanThiranVetterli:WhereYouAreIsWhoYouAreUserIdentificationByMatchingStatistics}) or semiconductor manufacturing (see, \eg, \textsc{Xu} and \textsc{Chu}~\cite{XuChu:AMatchingBasedDecomposerForDoublePatterningLithography}); it can also be utilised in a preprocessing step, which improves pivoting in solvers for large sparse linear systems (see, \eg, \textsc{Schenk} and \textsc{Gärtner}~\cite{SchenkGartner:OnFastFactorizationPivotingMethodsForSparseSymmetricIndefiniteSystems}). Since the MWM is obviously equivalent to the so-called \emph{maximum cardinality minimum weight matching problem}, other practical applications---among others---can be found in \textsc{Wu} and \textsc{Zhong}~\cite{WuZhong:FusionBlossomFastMWPMDecodersForQEC} or \textsc{Lu} et al.~\cite{LuGreevyXuBeck:OptimalNonbipartiteMatchingAndItsStatisticalApplications}.

	It is solvable in polynomial time in the graph size (see, \eg, \cite{LovaszPlummer:MatchingTheory,Schrijver:CombinatorialOptimizationPolyhedraAndEfficiency,KorteVygen:CombinatorialOptimizationTheoryAndAlgorithms,Galil:EfficientAlgorithmsForFindingMaximumMatchingInGraphs}) -- the famous \emph{Edmonds’ algorithm}~\cite{Edmonds:MaximumMatchingAndAPolyhedronWithO1Vertices} and its implementations provided by \textsc{Gabow}~\cite{Gabow:AnEfficientImplementationOfEdmondsAlgorithmForMaximumMatchingOnGraphs} and \textsc{Lawler}~\cite{Lawler:CombinatorialOptimizationNetworksAndMatroids} guarantee an $O(n^3)$ running time.
%; the first implementation achieving it was then provided by both \textsc{Gabow}~\cite{Gabow:AnEfficientImplementationOfEdmondsAlgorithmForMaximumMatchingOnGraphs} and \textsc{Lawler}~\cite{Lawler:CombinatorialOptimizationNetworksAndMatroids}.
	The fastest algorithm for the MWM, which guarantees an $O(n^{2.5})$ running time, was finally provided by \textsc{Micali} and \textsc{Vazirani}~\cite{MicaliVazirani:An0SqrtVEAlgoithmForFindingMaximumMatchingInGeneralGraphs}; the authors, however, do not prove the optimality of their algorithm in this paper; a complete proof was provided later by \textsc{Vazirani}~\cite{Vazirani:MaximumMatchingAndAPolyhedronWith01Vertices}. An efficient implementation capable of solving instances with millions of vertices to optimality was provided by \textsc{Kolmogorov}~\cite{Kolmogorov:BlossomVANewImplementationOfAMinimumCostPerfectMatchingAlgorithm};  although this algorithm performs effectively in practice, its worst-case complexity
%(which the author “believes to be $O(n^3 m)$”)
	does not improve the previous results.
%by \textsc{Gabow}~\cite{Gabow:AnEfficientImplementationOfEdmondsAlgorithmForMaximumMatchingOnGraphs} and \textsc{Lawler}~\cite{Lawler:CombinatorialOptimizationNetworksAndMatroids}.
	There exist various approximation algorithms with a fixed approximation ratio running in $O(m)$ time (see, \eg, \textsc{Preis}~\cite{Preis:LinearTime12ApproximationAlgorithmForMaximumWeightedMatchingInGeneralGraphs} and \textsc{Drake} and \textsc{Hougardy}~\cite{DrakeHougardy:ASimpleApproximationAlgorithmForTheWeightedMatchingProblem}). Further, \textsc{Duan} and \textsc{Pettie}~\cite{DuanPettie:LinearTimeApproximationForMaximumWeightMatching} present an FPTAS that finds a $(1 - \epsilon)$-approximate maximum weight matching in $O(m \epsilon^{-1} \log \epsilon^{-1})$ time.
	
	Given a finite set of $n$ distinct points $P = \{p_1, p_2, \ldots, p_n\} \subset \rz^2$, the \textbf{Euclidean maximum weight matching problem} (\textbf{Euclidean MWM}) asks for a matching $M \subseteq \binom{P}{2}$ that maximises the total Euclidean distance $f(P, M) \defeq \sum_{{u, v} \in M} \norm{v - u}_2$, where $\lVert \cdot \rVert_2$ denotes the {\em Euclidean norm}. Since $\lVert v - u \rVert_2 > 0$ for each $u \neq v \in P$, such an optimum matching is obviously \emph{perfect} for $|P|$ even and \emph{near-perfect} for $|P|$ odd; \ie\ it covers all vertices in the former and all but one vertex in the latter case. Obviously, this problem is a special case of the maximum weight matching problem. So all the algorithms mentioned above can be applied directly. Several heuristics applicable to the Euclidean MWM are known. \textsc{Avis}~\cite{Avis:ASurveyOfHeuristicsForTheWeightedMatchingProblem} provides a survey of various heuristics for the MWM and some of its special cases, which are faster than exact algorithms, but usually yield solutions much worse than the optimum. \textsc{Wu} and \textsc{Li}~\cite{WuLi:SolvingMaximumWeightedMatchingOnLargeGraphsWithDeepReinforcementLearning} use a deep reinforcement learning (DRL) model to find high-quality solutions for the general MWM. Finally, \textsc{Baumann} et al.~\cite{BaumannGoldschmidtHochbaum:AFastAlgorithmForEuclideanMaximumWeightNonBipartiteMatching} address the Euclidean MWM in a $d$-dimensional space and introduce a fast heuristic algorithm  yielding maximum weight matchings of high quality, often within less than 1\% of the optimum. To the best of our knowledge---and opposite to the minimisation case---neither optimum nor approximation algorithms guaranteeing a better worst-case running time are known for the Euclidean MWM.
	
	We provide a novel algorithm for the $2$-dimensional Euclidean MWM that guarantees both a worst-case computational time of $O(n \log n)$ and an excellent\footnote{In fact, we conjecture that the solution produced by the algorithm approaches the optimum as $n$ tends to infinity, \ie, that the algorithm is asymptotically optimal.} solution quality. Moreover, the presented algorithm proves to be extremely efficient from the practical point of view.

%	The remainder of this paper is organised as follows. In Section~\ref{section:ourAlgorithm}, we describe our algorithm, followed by computational results in Section~\ref{section:computationalResults}. Final notes, conclusions and outlook are finally provided in Section~\ref{section:finalNotesConclusionsAndOutlook}.

\section{Our algorithm}
	\label{section:ourAlgorithm}
	To solve the Euclidean MWM, we utilise an optimum algorithm for the \textbf{maximum angular-metric travelling salesperson problem}\footnote{In the context of the \emph{travelling salesperson problem} (and related problems like the \textbf{maximum angular-metric travelling salesperson problem}), “salesman” is sometimes used instead of “salesperson”.} (\textbf{MaxAngleTSP}) first introduced by \textsc{Aichholzer} et al.~\cite{AichholzerFischerFischerMeierPferschyPilzStanek:MinimizationAndMaximizationVersionsOfTheQuadraticTravellingSalesmanProblem}. Our notation is based on the notation used in this paper.
	
	Let $\widehat{G} = (\widehat{V}, \widehat{E})$ with $\widehat{V} = \{1, 2, \ldots, \widehat{n}\}$ and $\widehat{E} = \big\{\{i, j\} : i, j \in \widehat{V}, i \neq j\big\}$ be a complete undirected graph with $\widehat{n}$ vertices. We define a $2$-{\em edge} $\widehat{e}^{\la3\ra} \defeq \la i, j, k \ra$, where $i, j, k \in \widehat{V}$, $i \neq j$, $i \neq k$, and $j \neq k$ as a sequence of three distinct vertices where the reverse sequence is regarded as identical; \ie\ for each $i, j, k \in \widehat{V}$, $i \neq j$, $i \neq k$, and $j \neq k$, $\la i, j, k \ra = \la k, j, i \ra$. If there is no danger of confusion, we simply write $i j k$ instead of $\la i, j, k \ra$. Furthermore, we define a {\em complete $2$-graph} $\widehat{G}^{\la3\ra} = (\widehat{V}, \widehat{E}^{\la3\ra})$ as a pair of a vertex set $\widehat{V}$ and a set $\widehat{E}^{\la3\ra}$ of all possible $2$-edges in $\widehat{G}^{\la3\ra}$. Finally, a {\em tour} $\widehat{T} = \big(\sigma(1), \sigma(2), \ldots, \sigma(\widehat{n})\big)$ in a graph or $2$-graph with $\widehat{n}$ vertices is a permutation $\sigma$ of its vertices $1$, $2$, \ldots, $\widehat{n}$.
	
	Let $\widehat{G}^{\la3\ra} = (\widehat{V}, \widehat{E}^{\la3\ra})$ be a complete $2$-graph with $\widehat{n} \geq 3$ vertices corresponding to points in the Euclidean plane. We call this graph \emph{Euclidean} and define the so-called \textbf{inner angles} $\widehat{\alpha}_{i j k}$
	%(see Figure~\ref{figure:alphaWidehatAlpha})
	for each $2$-edge $\la i, j, k \ra \in \widehat{E}^{\la3\ra}$ in $\widehat{G}^{\la3\ra}$ as
	\begin{equation}
		\label{equation:innerAngles}
		\widehat{\alpha}_{i j k} \defeq \pi - \arccos_{[0, \pi]}{\left(\frac{j - i}{\norm{j - i}} \cdot \frac{k - j}{\norm{k - j}}\right)},
	\end{equation}
	where the dot $\cdot$ denotes the scalar product. The MaxAngleTSP asks for a tour $\widehat{T}$ \underline{minimising}\footnote{Although, we minimise the objective function value, we use the name \emph{maximum angular-metric TSP} to be consistent with other papers addressing this problem, which originates in the so-called \emph{symmetric travelling salesperson problem} -- for more details see \textsc{Aichholzer} et al.~\cite{AichholzerFischerFischerMeierPferschyPilzStanek:MinimizationAndMaximizationVersionsOfTheQuadraticTravellingSalesmanProblem}.} the sum of all its \emph{inner angles}
	\begin{equation}
		\label{equation:objectiveValue}
		\widehat{f}(\widehat{G}^{\la3\ra}, \widehat{T}) \defeq \left(\sum_{i = 1}^{\widehat{n} - 2}{\widehat{\alpha}_{\sigma(i) \sigma(i + 1) \sigma(i + 2)}}\right) + \widehat{\alpha}_{\sigma(\widehat{n} - 1) \sigma(\widehat{n}) \sigma(1)} + \widehat{\alpha}_{\sigma(\widehat{n}) \sigma(1) \sigma(2)}.
	\end{equation}
%	\begin{figure}[htb!]
%		\centering
%		\begin{tikzpicture}[ycomb, scale=0.7]
%			\node[circle, draw=black!100, fill=black!100, thick, inner sep=0pt, minimum size=0.5mm, label=left:{\color{black} $i$}] (nodeI) at (0.0, 0.0) {};
%			\node[circle, draw=black!100, fill=black!100, thick, inner sep=0pt, minimum size=0.5mm, label=above:{\color{black} $j$}, label={[yshift=-4pt]below:{\color{black} $\widehat{\alpha}_{i j k}$}}] (nodeJ) at (1, 1.0) {};
%			\node[circle, draw=black!100, fill=black!100, thick, inner sep=0pt, minimum size=0.5mm, label=right:{\color{black} $k$}] (nodeK) at (2.0, 0.0) {};
%			
%			\draw [help lines, dashed, shorten >= -0.8cm] (nodeI) -- (nodeJ);
%			\draw [help lines, dashed, shorten >= -0.8cm] (nodeK) -- (nodeJ);
%			\draw (nodeI) -- (nodeJ);
%			\draw (nodeK) -- (nodeJ);
%			
%			\draw (0.25, 0.25) arc (233.13:306.87:1.25);
%		\end{tikzpicture}
%		\caption{Illustration of the {\em inner angle} $\widehat{\alpha}_{i j k}$.}
%		\label{figure:alphaWidehatAlpha}
%	\end{figure}

	\textsc{Aichholzer} et al.~\cite{AichholzerFischerFischerMeierPferschyPilzStanek:MinimizationAndMaximizationVersionsOfTheQuadraticTravellingSalesmanProblem} found a polynomial time algorithm (in the graph size) for solving the MaxAngleTSP if $\widehat{n}$ is odd and proved the following theorem.
	\begin{theorem}[\textsc{Aichholzer} et al.~\cite{AichholzerFischerFischerMeierPferschyPilzStanek:MinimizationAndMaximizationVersionsOfTheQuadraticTravellingSalesmanProblem}]
		\label{theorem:timeComplexityMaxAngleTSP}
		Let $\widehat{G}^{\la3\ra} = (\widehat{V}, \widehat{E}^{\la3\ra})$ be a Euclidean complete $2$-graph with $\widehat{n} \geq 3$ being odd. Then an optimum solution $\widehat{T}$ to the MaxAngleTSP can be constructed in $O(n \log n)$ time; $\widehat{f}(\widehat{G}^{\la3\ra}, \widehat{T}) = \pi$.
	\end{theorem}
%	\begin{proof}
%		See \textsc{Aichholzer} et al.~\cite{AichholzerFischerFischerMeierPferschyPilzStanek:MinimizationAndMaximizationVersionsOfTheQuadraticTravellingSalesmanProblem}.
%	\end{proof}
	To the best of our knowledge, the complexity of the MaxAngleTSP in the case that $\widehat{n}$ is even remains an open question.

	Now, we can describe our algorithm (see also Algorithm~\ref{algorithm}).
	\begin{algorithm}[htb!]
		\begin{algorithmic}[1]
			\Require set of $n$ distinct points $P \subset \rz^2$ defining a complete graph $G = (P, E)$
			\Ensure matching $M \subseteq E$ that is perfect if $n$ is even; near-perfect if $n$ is odd
			\If{$n$ is odd}
			\label{algorithm:nOdd}
				\State find an optimum solution $\widehat{T} = \big(\sigma(1), \sigma(2), \ldots, \sigma(n)\big)$ of the correspond-\Statex[1]ing MaxAngleTSP instance;
				\label{algorithm:OptimumSolutionOfMaxAngleTSPIfNOdd}
				\State $M_1 \defeq \big\{\{p_{\sigma(2 i)}, p_{\sigma(2 i + 1)}\} | 1 \leq i \leq \left\lfloor\frac{n}{2}\right\rfloor\big\}$;
				\label{algorithm:createFirstM1}
				\State $f_{M_1} \defeq f(P, M_1)$;
				\label{algorithm:computeOVOfFirstM1}
				\State $M_2 \defeq \big\{\{p_{\sigma(2 i - 1)}, p_{\sigma(2 i)}\} | 2 \leq i \leq \left\lfloor\frac{n}{2}\right\rfloor\big\} \bigcup \big\{\{p_{\sigma(1)}, p_{\sigma(n)}\}\big\}$;
				\label{algorithm:createFirstM2}
				\State $f_{M_2} \defeq f(P, M_2)$;
				\label{algorithm:computeOVOfFirstM2}
				\If{$f(P, M_1) > f(P, M_2)$}
				\label{algorithm:compareFirstM1AndM2Then}
					\State $M^* \defeq M_1$; $f_{M^*} \defeq f_{M_1}$;
					\label{algorithm:setMAsteriskToM1}
				\Else
				\label{algorithm:compareFirstM1AndM2Else}
					\State $M^* \defeq M_2$; $f_{M^*} \defeq f_{M_2}$;
					\label{algorithm:setMAsteriskToM2}
				\EndIf
				\label{algorithm:compareFirstM1AndM2EndIf}
				\For{$i = 2, 4, \ldots, n - 1$}
				\label{algorithm:forI}
					\State $M_1 \defeq M_1 \setminus \big\{\{p_{\sigma(i)}, p_{\sigma(i + 1)}\}\big\} \bigcup \big\{\{p_{\sigma(i - 1)}, p_{\sigma(i)}\}\big\}$;
					\label{algorithm:computeNewM1}
					\State $f_{M_1} \defeq f_{M_1} - \norm{p_{\sigma(i)} - p_{\sigma(i + 1)}}_2 + \norm{p_{\sigma(i - 1)} - p_{\sigma(i)}}_2$;
					\label{algorithm:computeNewOVOfFirstM1}
					\If{$f_{M_1} > f_{M^*}$}
					\label{algorithm:ifNewM1IsBetterThen}
						\State $M^* \defeq M_1$; $f_{M^*} \defeq f_{M_1}$;
						\label{algorithm:setMAsteriskToNewM1}
					\EndIf
					\label{algorithm:ifNewM1IsBetterEndIf}
					\If{$i \neq n - 1$}
					\label{algorithm:ifIDoesNotEqualToNMinus1Then}
						\State $M_2 \defeq M_2 \setminus \big\{\{p_{\sigma(i + 1)}, p_{\sigma(i + 2)}\}\big\} \bigcup \big\{\{p_{\sigma(i)}, p_{\sigma(i + 1)}\}\big\}$;
						\label{algorithm:computeNewM2}
						\State $f_{M_2} \defeq f_{M_2} - \norm{p_{\sigma(i + 1)} - p_{\sigma(i + 2)}}_2 + \norm{p_{\sigma(i)} - p_{\sigma(i + 1)}}_2$;
						\label{algorithm:computeNewOVOfFirstM2}
						\If{$f_{M_2} > f_{M^*}$}
						\label{algorithm:ifNewM2IsBetterThen}
							\State $M^* \defeq M_2$; $f_{M^*} \defeq f_{M_2}$;
							\label{algorithm:setMAsteriskToNewM2}
						\EndIf
						\label{algorithm:ifNewM2IsBetterEndIf}
					\EndIf
					\label{algorithm:ifIDoesNotEqualToNMinus1EndIf}
				\EndFor
				\label{algorithm:endForI}
			\Else
			\label{algorithm:nEven}
				\State create a MaxAngleTSP instance by using all points in $P$ and adding one\Statex[1]additional point $p_{n + 1}$;
				\label{algorithm:createMaxAngleTSPInstanceIfNEven}
				\State find an optimum solution $\widehat{T} = \big(\sigma(1), \sigma(2), \ldots, \sigma(n + 1)\big)$;
				\label{algorithm:OptimumSolutionOfMaxAngleTSPIfNEven}
				\If{$\sum_{i = 1}^{\frac{n}{2}}\bigNorm{p_{\sigma(2 i)} - p_{\sigma(2 i - 1)}}_2 > \bigNorm{p_{\sigma(n)} - p_{\sigma(1)}}_2 + \sum_{i = 1}^{\frac{n}{2} - 1}\bigNorm{p_{\sigma(2 i + 1)} - p_{\sigma(2 i)}}_2$\Statex[1]\hspace{-1.3mm}}
				\label{algorithm:eachOddEdgeBetterIfNEven}
					\State $M^* \defeq \big\{\{p_{\sigma(2 i - 1)}, p_{\sigma(2 i)}\} | 1 \leq i \leq \frac{n}{2}\big\}$;
					\label{algorithm:returnEachOddEdgeBetterIfNEven}
				\Else
				\label{algorithm:eachEvenEdgeBetterIfNEven}
					\State $M^* \defeq \big\{\{p_{\sigma(2 i)}, p_{\sigma(2 i + 1)}\} | 1 \leq i \leq \frac{n}{2} - 1\big\} \bigcup \big\{\{p_{\sigma(1)}, p_{\sigma(n)}\}\big\}$;
					\label{algorithm:returnEachEvenEdgeBetterIfNEven}
				\EndIf
				\label{algorithm:eachOddEvenEdgeBetterIfNEvenEndIf}
			\EndIf
			\label{algorithm:nOddNEvenEndIf}
			\Return $M^*$;
			\label{algorithm:return}
		\end{algorithmic}
		\caption{Our algorithm}
		\label{algorithm}
	\end{algorithm}
	The general idea of our algorithm is to create an optimum MaxAngleTSP tour (see pseudocode lines~\ref{algorithm:OptimumSolutionOfMaxAngleTSPIfNOdd} and \ref{algorithm:OptimumSolutionOfMaxAngleTSPIfNEven}) and then find the best matching obtained by alternatingly taking and not taking the edges in this tour (see pseudocode lines \ref{algorithm:createFirstM1}--\ref{algorithm:endForI} and \ref{algorithm:eachOddEdgeBetterIfNEven}--\ref{algorithm:eachOddEvenEdgeBetterIfNEvenEndIf}).
	
	Since the algorithm by \textsc{Aichholzer} et al.~\cite{AichholzerFischerFischerMeierPferschyPilzStanek:MinimizationAndMaximizationVersionsOfTheQuadraticTravellingSalesmanProblem} can find an optimum MaxAngleTSP only for $n$ odd, we distinguish between the \emph{odd} case ($n$ odd) and the \emph{even} case ($n$ even). In the odd case, we can use the algorithm of \textsc{Aichholzer} et al.~\cite{AichholzerFischerFischerMeierPferschyPilzStanek:MinimizationAndMaximizationVersionsOfTheQuadraticTravellingSalesmanProblem} directly to obtain an optimum MaxAngleTSP tour (see pseudocode line~\ref{algorithm:OptimumSolutionOfMaxAngleTSPIfNOdd}). To create a matching by taking every second edge in this tour, we must first choose the vertex that remains unmatched (because only nearly-perfect matchings exist), and then use every second edge for the matching. Among all $n$ possibilities, we select the one with the maximum total weight of the matching edges (see pseudocode lines \ref{algorithm:createFirstM1}--\ref{algorithm:endForI}). In the even case, we first add an auxiliary, $(n + 1)$th, point in pseudocode line~\ref{algorithm:createMaxAngleTSPInstanceIfNEven} and then solve the corresponding MaxAngleTSP instance in pseudocode line~\ref{algorithm:OptimumSolutionOfMaxAngleTSPIfNEven}; for the creation of the resulting matching, the auxiliary point (which has the index $n + 1$) is ignored. In this case, we have two possibilities for creating a matching: either we start by taking the first or second edge, and then take every second edge in the MaxAngleTSP tour (while ignoring the auxiliary vertex $n + 1$). Again, the better possibility is chosen (see pseudocode lines~\ref{algorithm:eachOddEdgeBetterIfNEven}--\ref{algorithm:eachOddEvenEdgeBetterIfNEvenEndIf}).
	
	\begin{theorem}
		Algorithm~\ref{algorithm} runs in $O(n \log n)$ time.
	\end{theorem}
	\begin{proof}
		First, an optimum MaxAngleTSP tour is constructed in pseudocode lines \ref{algorithm:OptimumSolutionOfMaxAngleTSPIfNOdd} and \ref{algorithm:OptimumSolutionOfMaxAngleTSPIfNEven}; this can be done in $O(n \log n)$ time (see Theorem~\ref{theorem:timeComplexityMaxAngleTSP}). If $n$ is odd, first two matchings are created by leaving the first and the second vertex, respectively, unmatched and by pairing the remaining vertices in the order of these tours (see lines \ref{algorithm:createFirstM1}--\ref{algorithm:compareFirstM1AndM2EndIf}); this can be done in $O(2 n)$ time. Then, the unmatched vertices of these two matchings are shifted along the tour in pseudocode lines \ref{algorithm:forI}--\ref{algorithm:endForI} to obtain the remaining $n - 2$ possible matchings and their respective objective function values, which are not created and evaluated anew, but are obtained by updating the already existing matchings and computed objective function values, respectively, in pseudocode lines \ref{algorithm:computeNewM1}--\ref{algorithm:ifIDoesNotEqualToNMinus1EndIf}; this can be done in a constant time. If $n$ is even, there are just two possibilities for creating a matching in the order of the MaxAngleTSP tour. So, the overall running time is $O(n \log n + 2 n) = O(n \log n)$ for the odd case and $O(n \log n)$ for the even case.
	\end{proof}

\section{Computational results}
	\label{section:computationalResults}
	\subsection{Benchmark instances and test environment}
		\label{subsection:benchmarkInstancesAndTestEnvironment}
		We tested our algorithm with three types of test instances:
%        \begin{description}[style = multiline, leftmargin = 0.5cm, labelsep = 1.6cm]
		\begin{description}[wide, itemsep = 0.075cm, topsep = 0.15cm]
			\item[Square] test instances are based on points uniformly distributed in $[0, 100)^2$. We created $10$ instances for each $n = 100, 101, 200, 201, 300, 301, \ldots, 1000, 1001$, where each instance with $n$ even has a pendant with $n + 1$ vertices, which is created by just adding one additional random (uniformly distributed) point out of $[0, 100)^2$ (\ie\ these instances differ by just one point).
			\item[Circle] test instances are created in the same way as the \emph{square}, but the points are uniformly distributed in a circle with a diameter of $100$ (boundaries excluded).
			\item[$\mbox{TSPLIB}$] test instances are taken from the TSPLIB created by \textsc{Reinelt}~\cite{Reinelt:TSPLIB}. We used all instances with \texttt{EDGE\_WEIGHT\_TYPE} set to \texttt{EUC\_2D}, where “weights are Euclidean distances in 2-D”\cite{Reinelt:TSPLIB}. These instances were also used to demonstrate the behaviour of this algorithm for points that are not uniformly distributed in a square or circle; in fact, the TSPLIB test instances are often structured and contain various patterns, such as, among others, mesh grids or collinear points. To have an equal number of test instances of this type for $n$ odd and even, we create a twin for each TSPLIB test instance by removing the last vertex.
		\end{description}
		
		All tests were run on an \emph{Intel Core i5-9300H} processor with 64~GB RAM under \emph{TUXEDO OS 24.04.3} (based on \emph{Ubuntu 24.04 LTS}, Linux kernel 6.17.0), and all programs were written in \emph{Python 3.12.3}. To find optimum Euclidean MWM solutions, we used the method \texttt{max\_weight\_matching}, which “is based on the ‘blossom’ method for finding augmenting paths and the ‘primal-dual’ method for finding a matching of maximum weight, both methods invented by Jack Edmonds \cite{Galil:EfficientAlgorithmsForFindingMaximumMatchingInGraphs}”, from the \emph{NetworkX} package, version \emph{3.5}~\cite{HagbergSchultSwart:ExploringNetworkStructureDynamicsAndFunctionUsingNetworkX}.
		
		Moreover, in order to guarantee the relative reproducibility of our computational results, we (i)~allowed no additional swap memory and (ii)~ran all tests separately without other user processes in the background.
		
	\subsection{Evaluation layout}
		\label{subsection:evaluationLayout}
		As a basis for comparisons, we computed optimum objective function values for all test instances using the \emph{NetworkX} package (see Section~\ref{subsection:benchmarkInstancesAndTestEnvironment} for more details). For an instance $P$, we denote an optimum matching $\overline{M}$ and its corresponding objective function value $\overline{f}(P) \defeq f(P, \overline{M})$. Using this, we define the \textbf{objective function value ratio}
		\begin{equation}
			\label{equation:ration}
			r^*(P) \defeq \frac{f(P, M^*)}{\overline{f}(P)},
		\end{equation}
		where $f(P, M^*)$ is the objective function value of the matching $M^*$ obtained from Algorithm~\ref{algorithm}. If there is no danger of confusion, we simply write $r^*$ instead of $r^*(P)$. Obviously, $0 < r^*(P) \leq 1$ for all $P \subset \rz^2$, where $|P| \geq 2$, and $u \neq v$ for each $u, v \in P$. $1 - r^*(P)$ expresses the relative gap between the objective function value of the matching output by our algorithm and the objective function value of an optimum matching; we will denote $1 - r^*(P)$ the \textbf{optimality gap}.
		
		Since we have $10$ instances of each type, we always report the {\em geometric mean ratio} values for all random instances $P$ of the same type (square or circle) and the same size $n$; we will denote such ratios by $\rg$.
		
		For the running times, we report the {\em arithmetic means} $\ta$, again always over all random instances of the same type (square, circle, or TSPLIB) and size $n$.
		
		For the TSPLIB instances, we cannot group them in this way, so we report the objective function value ratios and running times for all of them.
		
		Finally, we report the geometric mean ratios and arithmetic mean running times for all instances of the same type (square, circle, or TSPLIB) over all instance sizes.
	
    \subsection{Test results}
		\label{subsection:testResults}
		The results for all random test instances (\emph{square} and \emph{circle}) are summarised in Table~\ref{table:resultsForRandomTestInstances}. The first column, which reports the graph size (number of vertices $n$), is followed by columns containing the mean ratios $\rg$, the running times $\ta$ of our algorithm (\emph{A.~\ref{algorithm}}) and of the blossom algorithm (\emph{B.}) for both the square and circle test instances. The same results are then graphically visualised in Figures~\ref{figure:objectiveFunctionRatioRgForSquareTestInstances}, \ref{figure:objectiveFunctionRatioRgForCircleTestInstances}, and \ref{figure:runningTimeTaForSquareAndCircleTestInstances}.
		
%		\addtolength{\tabcolsep}{-3pt}
		\begin{table}[htb!]
			\centering
			\scriptsize
			\caption{results for random test instances \\[0.2cm] \begin{tabular}{rll}$\bullet$ & $\rg$ & mean objective function ratio \\ $\bullet$ & $\ta$ (A.~\ref{algorithm}) & mean running time of Algorithm~\ref{algorithm} in seconds \\ $\bullet$ & $\ta$ (B.) & mean running time of blossom algorithm in seconds\end{tabular}}
			\begin{tabular}{r | S[round-mode=places, round-precision=4] S[round-mode=places, round-precision=2, table-number-alignment=right] S[round-mode=places, round-precision=2, table-number-alignment=right] | S[round-mode=places, round-precision=4] S[round-mode=places, round-precision=2, table-number-alignment=right] S[round-mode=places, round-precision=2, table-number-alignment=right]}
				\toprule
				& \multicolumn{3}{c|}{square}
				& \multicolumn{3}{c}{circle} \\
				$n$ & {$\rg$} & {$\ta$ (A.~\ref{algorithm})} & {$\ta$ (B.)} & {$\rg$} & {$\ta$ (A.~\ref{algorithm})} & {$\ta$ (B.)} \\
				\midrule
				100 & 0.9996 & 0.025455 & 0.877091 & 0.9996 & 0.046249 & 0.807975 \\
				101 & 0.9998 & 0.023074 & 0.775571 & 0.9999 & 0.016098 & 0.750204 \\
				\bigstrut[t]
				200 & 0.9984 & 0.058965 & 7.171759 & 0.9997 & 0.089990 & 6.740872 \\
				201 & 0.9993 & 0.043157 & 6.208970 & 0.9997 & 0.033877 & 6.623921 \\
				\bigstrut[t]
				300 & 0.9997 & 0.086592 & 24.055092 & 0.9998 & 0.144531 & 24.470735 \\
				301 & 0.9999 & 0.068358 & 23.616385 & 0.9998 & 0.057300 & 25.447677 \\
				\bigstrut[t]
				400 & 0.9998 & 0.140336 & 60.359162 & 0.9994 & 0.180922 & 60.957473 \\
				401 & 0.9999 & 0.089319 & 61.520204 & 0.9999 & 0.087551 & 63.746707 \\
				\bigstrut[t]
				500 & 0.9998 & 0.152887 & 120.727300 & 0.9998 & 0.158512 & 119.178185 \\
				501 & 0.9999 & 0.119687 & 117.698312 & 0.9999 & 0.106844 & 130.194403 \\
				\bigstrut[t]
				600 & 0.9999 & 0.195601 & 212.945603 & 1.0000 & 0.189255 & 205.654286 \\
				601 & 1.0000 & 0.169453 & 206.189475 & 1.0000 & 0.131119 & 223.930750 \\
				\bigstrut[t]
				700 & 0.9999 & 0.253092 & 336.099382 & 0.9999 & 0.221132 & 330.008247 \\
				701 & 0.9999 & 0.202350 & 323.269649 & 1.0000 & 0.151355 & 356.751093 \\
				\bigstrut[t]
				800 & 1.0000 & 0.263991 & 488.745310 & 0.9999 & 0.287482 & 496.700447 \\
				801 & 1.0000 & 0.189475 & 491.016690 & 0.9999 & 0.188868 & 502.782039 \\
				\bigstrut[t]
				900 & 0.9999 & 0.309444 & 698.546250 & 1.0000 & 0.351273 & 707.447535 \\
				901 & 1.0000 & 0.221048 & 715.757696 & 1.0000 & 0.212929 & 716.585789 \\
				\bigstrut[t]
				1000 & 1.0000 & 0.355263 & 1048.102510 & 0.9999 & 0.358971 & 1030.307444 \\
				1001 & 1.0000 & 0.248849 & 1018.668589 & 0.9999 & 0.243152 & 978.331644 \\
				\bottomrule
			\end{tabular}
			\label{table:resultsForRandomTestInstances}
		\end{table}
%		\addtolength{\tabcolsep}{3pt}
		
		\begin{figure}[htb!]
			\centering
			\begin{tikzpicture}[xscale=\xscale, yscale=0.9*\yscaleRatioLarge]
				\pgfgettransformentries{\xscaleTikz}{\@tempa}{\@tempa}{\yscaleTikz}{\@tempa}{\@tempa}
				
				\draw[very thin, color=gray, xstep=10, ystep=1] (0, 0) grid (100, 18 );
				
				\def\crossSizeX{\crossSize / \xscaleTikz};
				\def\crossSizeY{\crossSize / \yscaleTikz};
				
				\def\crossOne{(-\crossSizeX,-\crossSizeY) -- (\crossSizeX,\crossSizeY) (-\crossSizeX,\crossSizeY) -- (\crossSizeX,-\crossSizeY)};
				\def\crossTwo{(-\crossSizeX,0) -- (\crossSizeX,0) (0,\crossSizeY) -- (0,-\crossSizeY)};
				\def\crossThree{(0,0) -- (0,\crossSizeY) (0,0) -- (\crossSizeX,-\crossSizeY) (0,0) -- (-\crossSizeX,-\crossSizeY)};
				\def\crossFour{(0,0) -- (0,-\crossSizeY) (0,0) -- (\crossSizeX,\crossSizeY) (0,0) -- (-\crossSizeX,\crossSizeY)};
				\def\crossFive{(0,0) -- (\crossSizeX,0) (0,0) -- (-\crossSizeX,\crossSizeY) (0,0) -- (-\crossSizeX,-\crossSizeY)};
				\def\crossSix{(0,0) -- (-\crossSizeX,0) (0,0) -- (\crossSizeX,\crossSizeY) (0,0) -- (\crossSizeX,-\crossSizeY)};
				
				%line 1
				\draw[black, shift={( 10, 14 )}] \crossOne;
				\draw[black, shift={( 20, 2 )}] \crossOne;
				\draw[black, shift={( 30, 15 )}] \crossOne;
				\draw[black, shift={( 40, 16 )}] \crossOne;
				\draw[black, shift={( 50, 16 )}] \crossOne;
				\draw[black, shift={( 60, 17 )}] \crossOne;
				\draw[black, shift={( 70, 17 )}] \crossOne;
				\draw[black, shift={( 80, 18 )}] \crossOne;
				\draw[black, shift={( 90, 17 )}] \crossOne;
				\draw[black, shift={( 100, 18 )}] \crossOne;
				
				\draw[black!75, dotted, thick] ( 10, 14 ) -- ( 20, 2 ) -- ( 30, 15 ) -- ( 40, 16 ) -- ( 50, 16 ) -- ( 60, 17 ) -- ( 70, 17 ) -- ( 80, 18 ) -- ( 90, 17 ) -- ( 100, 18 );
				
				%line 2
				\draw[red, shift={( 10.1, 16 )}] \crossTwo;
				\draw[red, shift={( 20.1, 11 )}] \crossTwo;
				\draw[red, shift={( 30.1, 17 )}] \crossTwo;
				\draw[red, shift={( 40.1, 17 )}] \crossTwo;
				\draw[red, shift={( 50.1, 17 )}] \crossTwo;
				\draw[red, shift={( 60.1, 18 )}] \crossTwo;
				\draw[red, shift={( 70.1, 17 )}] \crossTwo;
				\draw[red, shift={( 80.1, 18 )}] \crossTwo;
				\draw[red, shift={( 90.1, 18 )}] \crossTwo;
				\draw[red, shift={( 100.1, 18 )}] \crossTwo;
				
				\draw[red!75, densely dotted, thick] ( 10.1, 16 ) -- ( 20.1, 11 ) -- ( 30.1, 17 ) -- ( 40.1, 17 ) -- ( 50.1, 17 ) -- ( 60.1, 18 ) -- ( 70.1, 17 ) -- ( 80.1, 18 ) -- ( 90.1, 18 ) -- ( 100.1, 18 );
				
				\def\axisAdditionalLengthPlusTikzX{\axisAdditionalLengthPlus / \xscaleTikz}
				\def\axisAdditionalLengthMinusTikzX{\axisAdditionalLengthMinus / \xscaleTikz}
				\draw[arrow] (-\axisAdditionalLengthMinusTikzX, 0) -- (100, 0) -- +(\axisAdditionalLengthPlusTikzX, 0) node[right] {\xAxis};
				\def\xTotalLengthPlus{100+\axisAdditionalLengthPlusTikzX}
				\draw (\xTotalLengthPlus, 0) node[right] {$\qquad$};
				\def\axisAdditionalLengthPlusTikzY{\axisAdditionalLengthPlus / \yscaleTikz}
				\def\axisAdditionalLengthMinusTikzY{\axisAdditionalLengthMinus / \yscaleTikz}
				\draw[arrow] (0, -\axisAdditionalLengthMinusTikzY) -- (0, 18 ) -- +(0, \axisAdditionalLengthPlusTikzY) node[above, yshift=-0.15cm] {\yAxisRatio};
				
				\def\axisLabelTikzY{\axisLabel / \yscaleTikz}
				\draw[shift={(0, 0)}] (0, \axisLabelTikzY) -- (0, -\axisLabelTikzY) node[below] {$0$};
				\foreach \pos in {10, 20, 30, 40, 50, 60, 70, 80, 90, 100} \draw[shift={(\pos, 0)}] (0, \axisLabelTikzY) -- (0, -\axisLabelTikzY) node[below] {$\pos0$};
				
				\def\axisLabelTikzX{\axisLabel / \xscaleTikz}
				\draw[shift={(0,0 )}] (\axisLabelTikzX, 0) -- (-\axisLabelTikzX, 0) node[left] {$0.9982$};
				%				\draw[shift={(0,1 )}] (\axisLabelTikzX, 0) -- (-\axisLabelTikzX, 0) node[left] {$0.9983$};
				\draw[shift={(0,2 )}] (\axisLabelTikzX, 0) -- (-\axisLabelTikzX, 0) node[left] {$0.9984$};
				%				\draw[shift={(0,3 )}] (\axisLabelTikzX, 0) -- (-\axisLabelTikzX, 0) node[left] {$0.9985$};
				\draw[shift={(0,4 )}] (\axisLabelTikzX, 0) -- (-\axisLabelTikzX, 0) node[left] {$0.9986$};
				%				\draw[shift={(0,5 )}] (\axisLabelTikzX, 0) -- (-\axisLabelTikzX, 0) node[left] {$0.9987$};
				\draw[shift={(0,6 )}] (\axisLabelTikzX, 0) -- (-\axisLabelTikzX, 0) node[left] {$0.9988$};
				%				\draw[shift={(0,7 )}] (\axisLabelTikzX, 0) -- (-\axisLabelTikzX, 0) node[left] {$0.9989$};
				\draw[shift={(0,8 )}] (\axisLabelTikzX, 0) -- (-\axisLabelTikzX, 0) node[left] {$0.9990$};
				%				\draw[shift={(0,9 )}] (\axisLabelTikzX, 0) -- (-\axisLabelTikzX, 0) node[left] {$0.9991$};
				\draw[shift={(0,10 )}] (\axisLabelTikzX, 0) -- (-\axisLabelTikzX, 0) node[left] {$0.9992$};
				
				%				\draw[shift={(0,11 )}] (\axisLabelTikzX, 0) -- (-\axisLabelTikzX, 0) node[left] {$0.9993$};
				\draw[shift={(0,12 )}] (\axisLabelTikzX, 0) -- (-\axisLabelTikzX, 0) node[left] {$0.9994$};
				%				\draw[shift={(0,13 )}] (\axisLabelTikzX, 0) -- (-\axisLabelTikzX, 0) node[left] {$0.9995$};
				\draw[shift={(0,14 )}] (\axisLabelTikzX, 0) -- (-\axisLabelTikzX, 0) node[left] {$0.9996$};
				%				\draw[shift={(0,15 )}] (\axisLabelTikzX, 0) -- (-\axisLabelTikzX, 0) node[left] {$0.9997$};
				\draw[shift={(0,16 )}] (\axisLabelTikzX, 0) -- (-\axisLabelTikzX, 0) node[left] {$0.9998$};
				%				\draw[shift={(0,17 )}] (\axisLabelTikzX, 0) -- (-\axisLabelTikzX, 0) node[left] {$0.9999$};
				\draw[shift={(0,18 )}] (\axisLabelTikzX, 0) -- (-\axisLabelTikzX, 0) node[left] {$1.0000$};
			\end{tikzpicture}
			\caption[objective function ratio $\rg$ for \emph{square} test instances]{objective function ratio $\rg$ for \emph{square} test instances \\ (
				\begin{tikzpicture}[xscale=\xscale, yscale=\yscaleRatioLarge]
					\pgfgettransformentries{\xscaleTikz}{\@tempa}{\@tempa}{\yscaleTikz}{\@tempa}{\@tempa}
					\def\crossSizeX{\crossSize / \xscaleTikz};
					\def\crossSizeY{\crossSize / \yscaleTikz};
					
					\def\crossOne{(-\crossSizeX,-\crossSizeY) -- (\crossSizeX,\crossSizeY) (-\crossSizeX,\crossSizeY) -- (\crossSizeX,-\crossSizeY)};
					\def\crossTwo{(-\crossSizeX,0) -- (\crossSizeX,0) (0,\crossSizeY) -- (0,-\crossSizeY)};
					
					\draw[black] \crossOne;
				\end{tikzpicture}
				for even and
				\begin{tikzpicture}[xscale=\xscale, yscale=\yscaleRatioLarge]
					\pgfgettransformentries{\xscaleTikz}{\@tempa}{\@tempa}{\yscaleTikz}{\@tempa}{\@tempa}
					\def\crossSizeX{\crossSize / \xscaleTikz};
					\def\crossSizeY{\crossSize / \yscaleTikz};
					
					\def\crossOne{(-\crossSizeX,-\crossSizeY) -- (\crossSizeX,\crossSizeY) (-\crossSizeX,\crossSizeY) -- (\crossSizeX,-\crossSizeY)};
					\def\crossTwo{(-\crossSizeX,0) -- (\crossSizeX,0) (0,\crossSizeY) -- (0,-\crossSizeY)};
					
					\draw[red] \crossTwo;
				\end{tikzpicture}
				for odd instance sizes $n$)
			}
			\label{figure:objectiveFunctionRatioRgForSquareTestInstances}
		\end{figure}
		
		\begin{figure}[htb!]
			\centering
			\begin{tikzpicture}[xscale=\xscale, yscale=0.9*\yscaleRatioLarge]
				\pgfgettransformentries{\xscaleTikz}{\@tempa}{\@tempa}{\yscaleTikz}{\@tempa}{\@tempa}
				
				\draw[very thin, color=gray, xstep=10, ystep=1] (0, 0) grid (100, 8 );
				
				\def\crossSizeX{\crossSize / \xscaleTikz};
				\def\crossSizeY{\crossSize / \yscaleTikz};
				
				\def\crossOne{(-\crossSizeX,-\crossSizeY) -- (\crossSizeX,\crossSizeY) (-\crossSizeX,\crossSizeY) -- (\crossSizeX,-\crossSizeY)};
				\def\crossTwo{(-\crossSizeX,0) -- (\crossSizeX,0) (0,\crossSizeY) -- (0,-\crossSizeY)};
				\def\crossThree{(0,0) -- (0,\crossSizeY) (0,0) -- (\crossSizeX,-\crossSizeY) (0,0) -- (-\crossSizeX,-\crossSizeY)};
				\def\crossFour{(0,0) -- (0,-\crossSizeY) (0,0) -- (\crossSizeX,\crossSizeY) (0,0) -- (-\crossSizeX,\crossSizeY)};
				\def\crossFive{(0,0) -- (\crossSizeX,0) (0,0) -- (-\crossSizeX,\crossSizeY) (0,0) -- (-\crossSizeX,-\crossSizeY)};
				\def\crossSix{(0,0) -- (-\crossSizeX,0) (0,0) -- (\crossSizeX,\crossSizeY) (0,0) -- (\crossSizeX,-\crossSizeY)};
				
				%line 1
				\draw[black, shift={( 10, 4 )}] \crossOne;
				\draw[black, shift={( 20, 5 )}] \crossOne;
				\draw[black, shift={( 30, 6 )}] \crossOne;
				\draw[black, shift={( 40, 2 )}] \crossOne;
				\draw[black, shift={( 50, 6 )}] \crossOne;
				\draw[black, shift={( 60, 8 )}] \crossOne;
				\draw[black, shift={( 70, 7 )}] \crossOne;
				\draw[black, shift={( 80, 7 )}] \crossOne;
				\draw[black, shift={( 90, 8 )}] \crossOne;
				\draw[black, shift={( 100, 7 )}] \crossOne;
				
				\draw[black!75, dotted, thick] ( 10, 4 ) -- ( 20, 5 ) -- ( 30, 6 ) -- ( 40, 2 ) -- ( 50, 6 ) -- ( 60, 8 ) -- ( 70, 7 ) -- ( 80, 7 ) -- ( 90, 8 ) -- ( 100, 7 );
				
				%line 2
				\draw[red, shift={( 10.1, 7 )}] \crossTwo;
				\draw[red, shift={( 20.1, 5 )}] \crossTwo;
				\draw[red, shift={( 30.1, 6 )}] \crossTwo;
				\draw[red, shift={( 40.1, 7 )}] \crossTwo;
				\draw[red, shift={( 50.1, 7 )}] \crossTwo;
				\draw[red, shift={( 60.1, 8 )}] \crossTwo;
				\draw[red, shift={( 70.1, 8 )}] \crossTwo;
				\draw[red, shift={( 80.1, 7 )}] \crossTwo;
				\draw[red, shift={( 90.1, 8 )}] \crossTwo;
				\draw[red, shift={( 100.1, 7 )}] \crossTwo;
				
				\draw[red!75, densely dotted, thick] ( 10.1, 7 ) -- ( 20.1, 5 ) -- ( 30.1, 6 ) -- ( 40.1, 7 ) -- ( 50.1, 7 ) -- ( 60.1, 8 ) -- ( 70.1, 8 ) -- ( 80.1, 7 ) -- ( 90.1, 8 ) -- ( 100.1, 7 );
				
				\def\axisAdditionalLengthPlusTikzX{\axisAdditionalLengthPlus / \xscaleTikz}
				\def\axisAdditionalLengthMinusTikzX{\axisAdditionalLengthMinus / \xscaleTikz}
				\draw[arrow] (-\axisAdditionalLengthMinusTikzX, 0) -- (100, 0) -- +(\axisAdditionalLengthPlusTikzX, 0) node[right] {\xAxis};
				\def\xTotalLengthPlus{100+\axisAdditionalLengthPlusTikzX}
				\draw (\xTotalLengthPlus, 0) node[right] {$\qquad$};
				\def\axisAdditionalLengthPlusTikzY{\axisAdditionalLengthPlus / \yscaleTikz}
				\def\axisAdditionalLengthMinusTikzY{\axisAdditionalLengthMinus / \yscaleTikz}
				\draw[arrow] (0, -\axisAdditionalLengthMinusTikzY) -- (0, 8 ) -- +(0, \axisAdditionalLengthPlusTikzY) node[above, yshift=-0.15cm] {\yAxisRatio};
				
				\def\axisLabelTikzY{\axisLabel / \yscaleTikz}
				\draw[shift={(0, 0)}] (0, \axisLabelTikzY) -- (0, -\axisLabelTikzY) node[below] {$0$};
				\foreach \pos in {10, 20, 30, 40, 50, 60, 70, 80, 90, 100} \draw[shift={(\pos, 0)}] (0, \axisLabelTikzY) -- (0, -\axisLabelTikzY) node[below] {$\pos0$};
				
				\def\axisLabelTikzX{\axisLabel / \xscaleTikz}
				\draw[shift={(0,0 )}] (\axisLabelTikzX, 0) -- (-\axisLabelTikzX, 0) node[left] {$0.9992$};
				%				\draw[shift={(0,1 )}] (\axisLabelTikzX, 0) -- (-\axisLabelTikzX, 0) node[left] {$0.9993$};
				\draw[shift={(0,2 )}] (\axisLabelTikzX, 0) -- (-\axisLabelTikzX, 0) node[left] {$0.9994$};
				%				\draw[shift={(0,3 )}] (\axisLabelTikzX, 0) -- (-\axisLabelTikzX, 0) node[left] {$0.9995$};
				\draw[shift={(0,4 )}] (\axisLabelTikzX, 0) -- (-\axisLabelTikzX, 0) node[left] {$0.9996$};
				%				\draw[shift={(0,5 )}] (\axisLabelTikzX, 0) -- (-\axisLabelTikzX, 0) node[left] {$0.9997$};
				\draw[shift={(0,6 )}] (\axisLabelTikzX, 0) -- (-\axisLabelTikzX, 0) node[left] {$0.9998$};
				%				\draw[shift={(0,7 )}] (\axisLabelTikzX, 0) -- (-\axisLabelTikzX, 0) node[left] {$0.9999$};
				\draw[shift={(0,8 )}] (\axisLabelTikzX, 0) -- (-\axisLabelTikzX, 0) node[left] {$1.0000$};
			\end{tikzpicture}
			\caption[objective function ratio $\rg$ for \emph{circle} test instances]{objective function ratio $\rg$ for \emph{circle} test instances \\ (
				\begin{tikzpicture}[xscale=\xscale, yscale=\yscaleRatioLarge]
					\pgfgettransformentries{\xscaleTikz}{\@tempa}{\@tempa}{\yscaleTikz}{\@tempa}{\@tempa}
					\def\crossSizeX{\crossSize / \xscaleTikz};
					\def\crossSizeY{\crossSize / \yscaleTikz};
					
					\def\crossOne{(-\crossSizeX,-\crossSizeY) -- (\crossSizeX,\crossSizeY) (-\crossSizeX,\crossSizeY) -- (\crossSizeX,-\crossSizeY)};
					\def\crossTwo{(-\crossSizeX,0) -- (\crossSizeX,0) (0,\crossSizeY) -- (0,-\crossSizeY)};
					
					\draw[black] \crossOne;
				\end{tikzpicture}
				for even and
				\begin{tikzpicture}[xscale=\xscale, yscale=\yscaleRatioLarge]
					\pgfgettransformentries{\xscaleTikz}{\@tempa}{\@tempa}{\yscaleTikz}{\@tempa}{\@tempa}
					\def\crossSizeX{\crossSize / \xscaleTikz};
					\def\crossSizeY{\crossSize / \yscaleTikz};
					
					\def\crossOne{(-\crossSizeX,-\crossSizeY) -- (\crossSizeX,\crossSizeY) (-\crossSizeX,\crossSizeY) -- (\crossSizeX,-\crossSizeY)};
					\def\crossTwo{(-\crossSizeX,0) -- (\crossSizeX,0) (0,\crossSizeY) -- (0,-\crossSizeY)};
					
					\draw[red] \crossTwo;
				\end{tikzpicture}
				for odd instance sizes $n$)
			}
			\label{figure:objectiveFunctionRatioRgForCircleTestInstances}
		\end{figure}
	
		\begin{figure}
			\begin{multicols}{2}
				\begin{tikzpicture}[xscale=0.45*\xscale, yscale=0.75*\yscaleTimeLarge]
					\pgfgettransformentries{\xscaleTikz}{\@tempa}{\@tempa}{\yscaleTikz}{\@tempa}{\@tempa}
					
					\draw[very thin, color=gray, xstep=10, ystep=10] (0, 0) grid (100, 100 );
					
					\def\crossSizeX{\crossSize / \xscaleTikz};
					\def\crossSizeY{\crossSize / \yscaleTikz};
					
					\def\crossOne{(-\crossSizeX,-\crossSizeY) -- (\crossSizeX,\crossSizeY) (-\crossSizeX,\crossSizeY) -- (\crossSizeX,-\crossSizeY)};
					\def\crossTwo{(-\crossSizeX,0) -- (\crossSizeX,0) (0,\crossSizeY) -- (0,-\crossSizeY)};
					\def\crossThree{(0,0) -- (0,\crossSizeY) (0,0) -- (\crossSizeX,-\crossSizeY) (0,0) -- (-\crossSizeX,-\crossSizeY)};
					\def\crossFour{(0,0) -- (0,-\crossSizeY) (0,0) -- (\crossSizeX,\crossSizeY) (0,0) -- (-\crossSizeX,\crossSizeY)};
					\def\crossFive{(0,0) -- (\crossSizeX,0) (0,0) -- (-\crossSizeX,\crossSizeY) (0,0) -- (-\crossSizeX,-\crossSizeY)};
					\def\crossSix{(0,0) -- (-\crossSizeX,0) (0,0) -- (\crossSizeX,\crossSizeY) (0,0) -- (\crossSizeX,-\crossSizeY)};
					
					%line 1
					\draw[black, shift={( 10, 0.003 )}] \crossOne;
					\draw[black, shift={( 20, 0.006 )}] \crossOne;
					\draw[black, shift={( 30, 0.009 )}] \crossOne;
					\draw[black, shift={( 40, 0.014 )}] \crossOne;
					\draw[black, shift={( 50, 0.015 )}] \crossOne;
					\draw[black, shift={( 60, 0.020 )}] \crossOne;
					\draw[black, shift={( 70, 0.025 )}] \crossOne;
					\draw[black, shift={( 80, 0.026 )}] \crossOne;
					\draw[black, shift={( 90, 0.031 )}] \crossOne;
					\draw[black, shift={( 100, 0.036 )}] \crossOne;
					
					\draw[black!75, dotted, thick] ( 10, 0.003 ) -- ( 20, 0.006 ) -- ( 30, 0.009 ) -- ( 40, 0.014 ) -- ( 50, 0.015 ) -- ( 60, 0.020 ) -- ( 70, 0.025 ) -- ( 80, 0.026 ) -- ( 90, 0.031 ) -- ( 100, 0.036 );
					
					%line 2
					\draw[red, shift={( 10.1, 0.002 )}] \crossTwo;
					\draw[red, shift={( 20.1, 0.004 )}] \crossTwo;
					\draw[red, shift={( 30.1, 0.007 )}] \crossTwo;
					\draw[red, shift={( 40.1, 0.009 )}] \crossTwo;
					\draw[red, shift={( 50.1, 0.012 )}] \crossTwo;
					\draw[red, shift={( 60.1, 0.017 )}] \crossTwo;
					\draw[red, shift={( 70.1, 0.020 )}] \crossTwo;
					\draw[red, shift={( 80.1, 0.019 )}] \crossTwo;
					\draw[red, shift={( 90.1, 0.022 )}] \crossTwo;
					\draw[red, shift={( 100.1, 0.025 )}] \crossTwo;
					
					\draw[red!75, densely dotted, thick] ( 10.1, 0.002 ) -- ( 20.1, 0.004 ) -- ( 30.1, 0.007 ) -- ( 40.1, 0.009 ) -- ( 50.1, 0.012 ) -- ( 60.1, 0.017 ) -- ( 70.1, 0.020 ) -- ( 80.1, 0.019 ) -- ( 90.1, 0.022 ) -- ( 100.1, 0.025 );
					
					%line 3
					\draw[black, shift={( 10, 0.088 )}] \crossThree;
					\draw[black, shift={( 20, 0.717 )}] \crossThree;
					\draw[black, shift={( 30, 2.406 )}] \crossThree;
					\draw[black, shift={( 40, 6.036 )}] \crossThree;
					\draw[black, shift={( 50, 12.073 )}] \crossThree;
					\draw[black, shift={( 60, 21.295 )}] \crossThree;
					\draw[black, shift={( 70, 33.610 )}] \crossThree;
					\draw[black, shift={( 80, 48.875 )}] \crossThree;
					\draw[black, shift={( 90, 69.855 )}] \crossThree;
					\draw[black, shift={( 100, 104.810 )}] \crossThree;
					
					\draw[black!75, dotted, thick] ( 10, 0.088 ) -- ( 20, 0.717 ) -- ( 30, 2.406 ) -- ( 40, 6.036 ) -- ( 50, 12.073 ) -- ( 60, 21.295 ) -- ( 70, 33.610 ) -- ( 80, 48.875 ) -- ( 90, 69.855 ) -- ( 100, 104.810 );
					
					%line 4
					\draw[red, shift={( 10.1, 0.078 )}] \crossFour;
					\draw[red, shift={( 20.1, 0.621 )}] \crossFour;
					\draw[red, shift={( 30.1, 2.362 )}] \crossFour;
					\draw[red, shift={( 40.1, 6.152 )}] \crossFour;
					\draw[red, shift={( 50.1, 11.770 )}] \crossFour;
					\draw[red, shift={( 60.1, 20.619 )}] \crossFour;
					\draw[red, shift={( 70.1, 32.327 )}] \crossFour;
					\draw[red, shift={( 80.1, 49.102 )}] \crossFour;
					\draw[red, shift={( 90.1, 71.576 )}] \crossFour;
					\draw[red, shift={( 100.1, 101.867 )}] \crossFour;
					
					\draw[red!75, densely dotted, thick] ( 10.1, 0.078 ) -- ( 20.1, 0.621 ) -- ( 30.1, 2.362 ) -- ( 40.1, 6.152 ) -- ( 50.1, 11.770 ) -- ( 60.1, 20.619 ) -- ( 70.1, 32.327 ) -- ( 80.1, 49.102 ) -- ( 90.1, 71.576 ) -- ( 100.1, 101.867 );
					
					\def\axisAdditionalLengthPlusTikzX{\axisAdditionalLengthPlus / \xscaleTikz}
					\def\axisAdditionalLengthMinusTikzX{\axisAdditionalLengthMinus / \xscaleTikz}
					\draw[arrow] (-\axisAdditionalLengthMinusTikzX, 0) -- (100, 0) -- +(\axisAdditionalLengthPlusTikzX, 0) node[right] {\xAxis};
					\def\xTotalLengthPlus{100+\axisAdditionalLengthPlusTikzX}
					\draw (\xTotalLengthPlus, 0) node[right] {$\qquad$};
					\def\axisAdditionalLengthPlusTikzY{\axisAdditionalLengthPlus / \yscaleTikz}
					\def\axisAdditionalLengthMinusTikzY{\axisAdditionalLengthMinus / \yscaleTikz}
					\draw[arrow] (0, -\axisAdditionalLengthMinusTikzY) -- (0, 100 ) -- +(0, \axisAdditionalLengthPlusTikzY) node[above, yshift=-0.15cm] {\yAxisTime};
					
					\def\axisLabelTikzY{\axisLabel / \yscaleTikz}
					\draw[shift={(0, 0)}] (0, \axisLabelTikzY) -- (0, -\axisLabelTikzY) node[below] {$0$};
					\foreach \pos in {20, 40, 60, 80, 100} \draw[shift={(\pos, 0)}] (0, \axisLabelTikzY) -- (0, -\axisLabelTikzY) node[below] {$\pos0$};
					
					\def\axisLabelTikzX{\axisLabel / \xscaleTikz}
					\foreach \pos in {20, 40, 60, 80, 100} \draw[shift={(0, \pos)}] (\axisLabelTikzX, 0) -- (-\axisLabelTikzX, 0) node[left] {$\pos0$};
				\end{tikzpicture}
				
				\begin{tikzpicture}[xscale=0.45*\xscale, yscale=0.75*\yscaleTimeLarge]
					\pgfgettransformentries{\xscaleTikz}{\@tempa}{\@tempa}{\yscaleTikz}{\@tempa}{\@tempa}
					
					\draw[very thin, color=gray, xstep=10, ystep=10] (0, 0) grid (100, 100 );
					
					\def\crossSizeX{\crossSize / \xscaleTikz};
					\def\crossSizeY{\crossSize / \yscaleTikz};
					
					\def\crossOne{(-\crossSizeX,-\crossSizeY) -- (\crossSizeX,\crossSizeY) (-\crossSizeX,\crossSizeY) -- (\crossSizeX,-\crossSizeY)};
					\def\crossTwo{(-\crossSizeX,0) -- (\crossSizeX,0) (0,\crossSizeY) -- (0,-\crossSizeY)};
					\def\crossThree{(0,0) -- (0,\crossSizeY) (0,0) -- (\crossSizeX,-\crossSizeY) (0,0) -- (-\crossSizeX,-\crossSizeY)};
					\def\crossFour{(0,0) -- (0,-\crossSizeY) (0,0) -- (\crossSizeX,\crossSizeY) (0,0) -- (-\crossSizeX,\crossSizeY)};
					\def\crossFive{(0,0) -- (\crossSizeX,0) (0,0) -- (-\crossSizeX,\crossSizeY) (0,0) -- (-\crossSizeX,-\crossSizeY)};
					\def\crossSix{(0,0) -- (-\crossSizeX,0) (0,0) -- (\crossSizeX,\crossSizeY) (0,0) -- (\crossSizeX,-\crossSizeY)};
					
					%line 1
					\draw[black, shift={( 10, 0.005 )}] \crossOne;
					\draw[black, shift={( 20, 0.009 )}] \crossOne;
					\draw[black, shift={( 30, 0.014 )}] \crossOne;
					\draw[black, shift={( 40, 0.018 )}] \crossOne;
					\draw[black, shift={( 50, 0.016 )}] \crossOne;
					\draw[black, shift={( 60, 0.019 )}] \crossOne;
					\draw[black, shift={( 70, 0.022 )}] \crossOne;
					\draw[black, shift={( 80, 0.029 )}] \crossOne;
					\draw[black, shift={( 90, 0.035 )}] \crossOne;
					\draw[black, shift={( 100, 0.036 )}] \crossOne;
					
					\draw[black!75, dotted, thick] ( 10, 0.005 ) -- ( 20, 0.009 ) -- ( 30, 0.014 ) -- ( 40, 0.018 ) -- ( 50, 0.016 ) -- ( 60, 0.019 ) -- ( 70, 0.022 ) -- ( 80, 0.029 ) -- ( 90, 0.035 ) -- ( 100, 0.036 );
					
					%line 2
					\draw[red, shift={( 10.1, 0.002 )}] \crossTwo;
					\draw[red, shift={( 20.1, 0.003 )}] \crossTwo;
					\draw[red, shift={( 30.1, 0.006 )}] \crossTwo;
					\draw[red, shift={( 40.1, 0.009 )}] \crossTwo;
					\draw[red, shift={( 50.1, 0.011 )}] \crossTwo;
					\draw[red, shift={( 60.1, 0.013 )}] \crossTwo;
					\draw[red, shift={( 70.1, 0.015 )}] \crossTwo;
					\draw[red, shift={( 80.1, 0.019 )}] \crossTwo;
					\draw[red, shift={( 90.1, 0.021 )}] \crossTwo;
					\draw[red, shift={( 100.1, 0.024 )}] \crossTwo;
					
					\draw[red!75, densely dotted, thick] ( 10.1, 0.002 ) -- ( 20.1, 0.003 ) -- ( 30.1, 0.006 ) -- ( 40.1, 0.009 ) -- ( 50.1, 0.011 ) -- ( 60.1, 0.013 ) -- ( 70.1, 0.015 ) -- ( 80.1, 0.019 ) -- ( 90.1, 0.021 ) -- ( 100.1, 0.024 );
					
					%line 3
					\draw[black, shift={( 10, 0.081 )}] \crossThree;
					\draw[black, shift={( 20, 0.674 )}] \crossThree;
					\draw[black, shift={( 30, 2.447 )}] \crossThree;
					\draw[black, shift={( 40, 6.096 )}] \crossThree;
					\draw[black, shift={( 50, 11.918 )}] \crossThree;
					\draw[black, shift={( 60, 20.565 )}] \crossThree;
					\draw[black, shift={( 70, 33.001 )}] \crossThree;
					\draw[black, shift={( 80, 49.670 )}] \crossThree;
					\draw[black, shift={( 90, 70.745 )}] \crossThree;
					\draw[black, shift={( 100, 103.031 )}] \crossThree;
					
					\draw[black!75, dotted, thick] ( 10, 0.081 ) -- ( 20, 0.674 ) -- ( 30, 2.447 ) -- ( 40, 6.096 ) -- ( 50, 11.918 ) -- ( 60, 20.565 ) -- ( 70, 33.001 ) -- ( 80, 49.670 ) -- ( 90, 70.745 ) -- ( 100, 103.031 );

					%line 4
					\draw[red, shift={( 10.1, 0.075 )}] \crossFour;
					\draw[red, shift={( 20.1, 0.662 )}] \crossFour;
					\draw[red, shift={( 30.1, 2.545 )}] \crossFour;
					\draw[red, shift={( 40.1, 6.375 )}] \crossFour;
					\draw[red, shift={( 50.1, 13.019 )}] \crossFour;
					\draw[red, shift={( 60.1, 22.393 )}] \crossFour;
					\draw[red, shift={( 70.1, 35.675 )}] \crossFour;
					\draw[red, shift={( 80.1, 50.278 )}] \crossFour;
					\draw[red, shift={( 90.1, 71.659 )}] \crossFour;
					\draw[red, shift={( 100.1, 97.833 )}] \crossFour;
					
					\draw[red!75, densely dotted, thick] ( 10.1, 0.075 ) -- ( 20.1, 0.662 ) -- ( 30.1, 2.545 ) -- ( 40.1, 6.375 ) -- ( 50.1, 13.019 ) -- ( 60.1, 22.393 ) -- ( 70.1, 35.675 ) -- ( 80.1, 50.278 ) -- ( 90.1, 71.659 ) -- ( 100.1, 97.833 );

					\def\axisAdditionalLengthPlusTikzX{\axisAdditionalLengthPlus / \xscaleTikz}
					\def\axisAdditionalLengthMinusTikzX{\axisAdditionalLengthMinus / \xscaleTikz}
					\draw[arrow] (-\axisAdditionalLengthMinusTikzX, 0) -- (100, 0) -- +(\axisAdditionalLengthPlusTikzX, 0) node[right] {\xAxis};
					\def\xTotalLengthPlus{100+\axisAdditionalLengthPlusTikzX}
					\draw (\xTotalLengthPlus, 0) node[right] {$\qquad$};
					\def\axisAdditionalLengthPlusTikzY{\axisAdditionalLengthPlus / \yscaleTikz}
					\def\axisAdditionalLengthMinusTikzY{\axisAdditionalLengthMinus / \yscaleTikz}
					\draw[arrow] (0, -\axisAdditionalLengthMinusTikzY) -- (0, 100 ) -- +(0, \axisAdditionalLengthPlusTikzY) node[above, yshift=-0.15cm] {\yAxisTime};
					
					\def\axisLabelTikzY{\axisLabel / \yscaleTikz}
					\draw[shift={(0, 0)}] (0, \axisLabelTikzY) -- (0, -\axisLabelTikzY) node[below] {$0$};
					\foreach \pos in {20, 40, 60, 80, 100} \draw[shift={(\pos, 0)}] (0, \axisLabelTikzY) -- (0, -\axisLabelTikzY) node[below] {$\pos0$};
					
					\def\axisLabelTikzX{\axisLabel / \xscaleTikz}
					\foreach \pos in {20, 40, 60, 80, 100} \draw[shift={(0, \pos)}] (\axisLabelTikzX, 0) -- (-\axisLabelTikzX, 0) node[left] {$\pos0$};
				\end{tikzpicture}
			\end{multicols}
			\caption[running time $\ta$ in seconds for \emph{square} (left) and \emph{circle} (right) test instances]{running time $\ta$ in seconds for \emph{square} (left) and \emph{circle} (right) test instances \\[0.2cm] \begin{tabular}{rll}$\bullet$ & Algorithm~\ref{algorithm}: &
				\begin{tikzpicture}[xscale=\xscale, yscale=\yscaleRatioLarge]
					\pgfgettransformentries{\xscaleTikz}{\@tempa}{\@tempa}{\yscaleTikz}{\@tempa}{\@tempa}
					\def\crossSizeX{\crossSize / \xscaleTikz};
					\def\crossSizeY{\crossSize / \yscaleTikz};
					
					\def\crossOne{(-\crossSizeX,-\crossSizeY) -- (\crossSizeX,\crossSizeY) (-\crossSizeX,\crossSizeY) -- (\crossSizeX,-\crossSizeY)};
					\def\crossTwo{(-\crossSizeX,0) -- (\crossSizeX,0) (0,\crossSizeY) -- (0,-\crossSizeY)};
					
					\draw[black] \crossOne;
				\end{tikzpicture}
				for even and
				\begin{tikzpicture}[xscale=\xscale, yscale=\yscaleRatioLarge]
					\pgfgettransformentries{\xscaleTikz}{\@tempa}{\@tempa}{\yscaleTikz}{\@tempa}{\@tempa}
					\def\crossSizeX{\crossSize / \xscaleTikz};
					\def\crossSizeY{\crossSize / \yscaleTikz};
					
					\def\crossOne{(-\crossSizeX,-\crossSizeY) -- (\crossSizeX,\crossSizeY) (-\crossSizeX,\crossSizeY) -- (\crossSizeX,-\crossSizeY)};
					\def\crossTwo{(-\crossSizeX,0) -- (\crossSizeX,0) (0,\crossSizeY) -- (0,-\crossSizeY)};
					
					\draw[red] \crossTwo;
				\end{tikzpicture}
				for odd instance sizes $n$ \\ $\bullet$ & blossom algorithm: & 
				\begin{tikzpicture}[xscale=\xscale, yscale=\yscaleRatioLarge]
					\pgfgettransformentries{\xscaleTikz}{\@tempa}{\@tempa}{\yscaleTikz}{\@tempa}{\@tempa}
					\def\crossSizeX{\crossSize / \xscaleTikz};
					\def\crossSizeY{\crossSize / \yscaleTikz};
					
					\def\crossThree{(0,0) -- (0,\crossSizeY) (0,0) -- (\crossSizeX,-\crossSizeY) (0,0) -- (-\crossSizeX,-\crossSizeY)};
					\def\crossFour{(0,0) -- (0,-\crossSizeY) (0,0) -- (\crossSizeX,\crossSizeY) (0,0) -- (-\crossSizeX,\crossSizeY)};
					
					\draw[black] \crossThree;
				\end{tikzpicture}
				for even and
				\begin{tikzpicture}[xscale=\xscale, yscale=\yscaleRatioLarge]
					\pgfgettransformentries{\xscaleTikz}{\@tempa}{\@tempa}{\yscaleTikz}{\@tempa}{\@tempa}
					\def\crossSizeX{\crossSize / \xscaleTikz};
					\def\crossSizeY{\crossSize / \yscaleTikz};
					
					\def\crossThree{(0,0) -- (0,\crossSizeY) (0,0) -- (\crossSizeX,-\crossSizeY) (0,0) -- (-\crossSizeX,-\crossSizeY)};
					\def\crossFour{(0,0) -- (0,-\crossSizeY) (0,0) -- (\crossSizeX,\crossSizeY) (0,0) -- (-\crossSizeX,\crossSizeY)};
					
					\draw[red] \crossFour;
				\end{tikzpicture}
				for odd instance sizes $n$\end{tabular}
			}
			\label{figure:runningTimeTaForSquareAndCircleTestInstances}
		\end{figure}
		
		The computational results demonstrate the excellent performance of Algorithm~\ref{algorithm} on both random test instance types (\emph{square} and \emph{circle}). Let us first focus on the mean ratios $\rg$ (see Table~\ref{table:resultsForRandomTestInstances} and Figures~\ref{figure:objectiveFunctionRatioRgForSquareTestInstances}, \ref{figure:objectiveFunctionRatioRgForCircleTestInstances} in particular): a tiny gap of less than $2$\textperthousand\ can be observed for smaller test instances ($n < 300$), where the algorithm performs slightly better for \emph{circle} test instances than for the \emph{square} ones. For the larger instances, the mean ratios are essentially equal to $1$, indicating that Algorithm~\ref{algorithm} consistently produces solutions that are very close to optimum and, in many cases, attain the optimum at least for points uniformly distributed in the Euclidean plane within a square or circle.

		We also tested whether the optimality gaps vary significantly with the exact position of the additional point temporarily introduced in pseudocode line~\ref{algorithm:createMaxAngleTSPInstanceIfNEven}. This was not the case. We tested various positions---\eg, their centroid or positions outside their convex hull (“far away”)---and small differences were observable only for very small test instances. For larger test instances, and especially asymptotically for $n$ tending to infinity, the exact position of this auxiliary point seems to play no role. A similar situation was observed regarding the particular point on the convex hull, which has to be chosen as the starting point in the optimum MaxAngleTSP algorithm introduced by \textsc{Aichholzer} et al.~\cite{AichholzerFischerFischerMeierPferschyPilzStanek:MinimizationAndMaximizationVersionsOfTheQuadraticTravellingSalesmanProblem} -- it plays just a negligible role for very small test instances; for larger test instances, no significant differences could be observed.
		
		To avoid focusing only on points randomly distributed in the Euclidean plane, we also tested our algorithm using structured test instances. For this purpose, we utilised the TSPLIB (see~\cite{Reinelt:TSPLIB}), which contains a broad range of structured instances, including Euclidean test instances exhibiting mesh-grid patterns and containing evident point clustering. As shown in Table~\ref{table:resultsForTSPLIBTestInstances},
%		in the Appendix,
		Algorithm~\ref{algorithm} achieves near-optimum solutions for these instances as well. In particular, we tested our algorithm for all Euclidean TSPLIB test instances with $n \leq 1084$, using each of them in its original form and without its last point to get the same number of test instances with an even and an odd number of vertices (see Section~\ref{subsection:benchmarkInstancesAndTestEnvironment} for more details); the worst objective function value ratio $\rg$ we obtained was $0.9862$ for the test instance \emph{fl417} in its original form (\ie\ using all $n = 417$ vertices), which corresponds to a “drilling problem”\cite{Reinelt:TSPLIB} and contains mesh-grid patterns.
		
		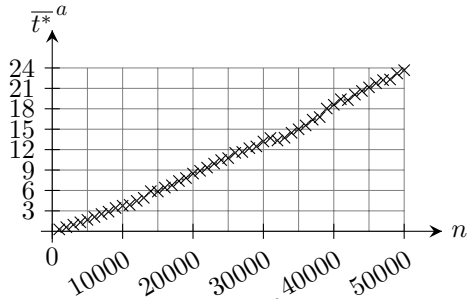
\begin{wrapfigure}{l}{0.5\textwidth}
			\centering
			\vspace*{-0.8cm}
			\begin{tikzpicture}[xscale=1*\xscale, yscale=1.8*\yscaleTimeLarge]
				\pgfgettransformentries{\xscaleTikz}{\@tempa}{\@tempa}{\yscaleTikz}{\@tempa}{\@tempa}
				
				\draw[very thin, color=gray, xstep=5, ystep=3] (0, 0) grid (50, 24 );
				
				\def\crossSizeX{\crossSize / \xscaleTikz};
				\def\crossSizeY{\crossSize / \yscaleTikz};
				
				\def\crossOne{(-\crossSizeX,-\crossSizeY) -- (\crossSizeX,\crossSizeY) (-\crossSizeX,\crossSizeY) -- (\crossSizeX,-\crossSizeY)};
				\def\crossTwo{(-\crossSizeX,0) -- (\crossSizeX,0) (0,\crossSizeY) -- (0,-\crossSizeY)};
				\def\crossThree{(0,0) -- (0,\crossSizeY) (0,0) -- (\crossSizeX,-\crossSizeY) (0,0) -- (-\crossSizeX,-\crossSizeY)};
				\def\crossFour{(0,0) -- (0,-\crossSizeY) (0,0) -- (\crossSizeX,\crossSizeY) (0,0) -- (-\crossSizeX,\crossSizeY)};
				\def\crossFive{(0,0) -- (\crossSizeX,0) (0,0) -- (-\crossSizeX,\crossSizeY) (0,0) -- (-\crossSizeX,-\crossSizeY)};
				\def\crossSix{(0,0) -- (-\crossSizeX,0) (0,0) -- (\crossSizeX,\crossSizeY) (0,0) -- (\crossSizeX,-\crossSizeY)};
				
				%line 1
				\draw[black, shift={( 1, 0.2781 )}] \crossOne;
				\draw[black, shift={( 2, 0.6084 )}] \crossOne;
				\draw[black, shift={( 3, 0.9703 )}] \crossOne;
				\draw[black, shift={( 4, 1.3267 )}] \crossOne;
				\draw[black, shift={( 5, 1.6844 )}] \crossOne;
				\draw[black, shift={( 6, 2.0962 )}] \crossOne;
				\draw[black, shift={( 7, 2.5399 )}] \crossOne;
				\draw[black, shift={( 8, 2.8848 )}] \crossOne;
				\draw[black, shift={( 9, 3.4264 )}] \crossOne;
				\draw[black, shift={( 10, 3.7552 )}] \crossOne;
				\draw[black, shift={( 11, 3.8781 )}] \crossOne;
				\draw[black, shift={( 12, 4.4712 )}] \crossOne;
				\draw[black, shift={( 13, 4.9148 )}] \crossOne;
				\draw[black, shift={( 14, 5.8961 )}] \crossOne;
				\draw[black, shift={( 15, 5.8348 )}] \crossOne;
				\draw[black, shift={( 16, 6.4216 )}] \crossOne;
				\draw[black, shift={( 17, 6.7624 )}] \crossOne;
				\draw[black, shift={( 18, 7.3879 )}] \crossOne;
				\draw[black, shift={( 19, 7.8280 )}] \crossOne;
				\draw[black, shift={( 20, 8.4796 )}] \crossOne;
				\draw[black, shift={( 21, 8.8894 )}] \crossOne;
				\draw[black, shift={( 22, 9.3614 )}] \crossOne;
				\draw[black, shift={( 23, 9.9439 )}] \crossOne;
				\draw[black, shift={( 24, 10.4668 )}] \crossOne;
				\draw[black, shift={( 25, 10.7514 )}] \crossOne;
				\draw[black, shift={( 26, 11.5220 )}] \crossOne;
				\draw[black, shift={( 27, 11.6556 )}] \crossOne;
				\draw[black, shift={( 28, 12.2129 )}] \crossOne;
				\draw[black, shift={( 29, 12.4733 )}] \crossOne;
				\draw[black, shift={( 30, 13.1957 )}] \crossOne;
				\draw[black, shift={( 31, 13.7073 )}] \crossOne;
				\draw[black, shift={( 32, 13.3365 )}] \crossOne;
				\draw[black, shift={( 33, 13.7406 )}] \crossOne;
				\draw[black, shift={( 34, 14.4678 )}] \crossOne;
				\draw[black, shift={( 35, 14.9900 )}] \crossOne;
				\draw[black, shift={( 36, 15.5512 )}] \crossOne;
				\draw[black, shift={( 37, 16.3951 )}] \crossOne;
				\draw[black, shift={( 38, 16.7787 )}] \crossOne;
				\draw[black, shift={( 39, 18.0420 )}] \crossOne;
				\draw[black, shift={( 40, 18.5893 )}] \crossOne;
				\draw[black, shift={( 41, 19.4021 )}] \crossOne;
				\draw[black, shift={( 42, 19.2641 )}] \crossOne;
				\draw[black, shift={( 43, 20.0752 )}] \crossOne;
				\draw[black, shift={( 44, 20.6202 )}] \crossOne;
				\draw[black, shift={( 45, 21.1472 )}] \crossOne;
				\draw[black, shift={( 46, 21.7537 )}] \crossOne;
				\draw[black, shift={( 47, 22.2143 )}] \crossOne;
				\draw[black, shift={( 48, 22.3060 )}] \crossOne;
				\draw[black, shift={( 49, 23.1917 )}] \crossOne;
				\draw[black, shift={( 50, 23.6724 )}] \crossOne;
				
				\draw[black!75, dotted, thick] ( 1, 0.2781 ) -- ( 2, 0.6084 ) -- ( 3, 0.9703 ) -- ( 4, 1.3267 ) --	( 5, 1.6844 ) -- ( 6, 2.0962 ) -- ( 7, 2.5399 ) -- ( 8, 2.8848 ) -- ( 9, 3.4264 ) -- ( 10, 3.7552 ) -- ( 11, 3.8781 ) -- ( 12, 4.4712 ) -- ( 13, 4.9148 ) -- ( 14, 5.8961 ) -- ( 15, 5.8348 ) -- ( 16, 6.4216 ) -- ( 17, 6.7624 ) -- ( 18, 7.3879 ) -- ( 19, 7.8280 ) -- ( 20, 8.4796 ) -- ( 21, 8.8894 ) -- ( 22, 9.3614 ) -- ( 23, 9.9439 ) -- ( 24, 10.4668 ) -- ( 25, 10.7514 ) -- ( 26, 11.5220 ) -- ( 27, 11.6556 ) -- ( 28, 12.2129 ) --
				( 29, 12.4733 ) -- ( 30, 13.1957 ) -- ( 31, 13.7073 ) -- ( 32, 13.3365 ) -- ( 33, 13.7406 ) -- ( 34, 14.4678 ) -- ( 35, 14.9900 ) -- ( 36, 15.5512 ) -- ( 37, 16.3951 ) -- ( 38, 16.7787 ) -- ( 39, 18.0420 ) -- ( 40, 18.5893 ) -- ( 41, 19.4021 ) -- ( 42, 19.2641 ) -- ( 43, 20.0752 ) -- ( 44, 20.6202 ) -- ( 45, 21.1472 ) -- ( 46, 21.7537 ) -- ( 47, 22.2143 ) -- ( 48, 22.3060 ) --	( 49, 23.1917 ) -- ( 50, 23.6724 );
				
				\def\axisAdditionalLengthPlusTikzX{\axisAdditionalLengthPlus / \xscaleTikz}
				\def\axisAdditionalLengthMinusTikzX{\axisAdditionalLengthMinus / \xscaleTikz}
				\draw[arrow] (-\axisAdditionalLengthMinusTikzX, 0) -- (50, 0) -- +(\axisAdditionalLengthPlusTikzX, 0) node[right] {\xAxis};
				\def\xTotalLengthPlus{100+\axisAdditionalLengthPlusTikzX}
				\draw (\xTotalLengthPlus, 0) node[right] {$\qquad$};
				\def\axisAdditionalLengthPlusTikzY{\axisAdditionalLengthPlus / \yscaleTikz}
				\def\axisAdditionalLengthMinusTikzY{\axisAdditionalLengthMinus / \yscaleTikz}
				\draw[arrow] (0, -\axisAdditionalLengthMinusTikzY) -- (0, 24 ) -- +(0, \axisAdditionalLengthPlusTikzY) node[above, yshift=-0.15cm] {\yAxisTime};
				
				\def\axisLabelTikzY{\axisLabel / \yscaleTikz}
				\draw[shift={(0, 0)}] (0, \axisLabelTikzY) -- (0, -\axisLabelTikzY) node[below] {$0$};
				\foreach \pos in {10, 20, 30, 40, 50} \draw[shift={(\pos, 0)}] (0, \axisLabelTikzY) -- (0, -\axisLabelTikzY) node[below, rotate around={30:(0, 0)}, shift={(-0.5, 0)}] {$\pos000$};
				
				\def\axisLabelTikzX{\axisLabel / \xscaleTikz}
				\foreach \pos in {3, 6, 9, 12, 15, 18, 21, 24} \draw[shift={(0, \pos)}] (\axisLabelTikzX, 0) -- (-\axisLabelTikzX, 0) node[left] {$\pos$};
			\end{tikzpicture}
			\vspace*{-1cm}
			\caption{running time $\ta$ in seconds for large \emph{square} test instances}
			\vspace*{-1cm}
			\label{figure:runningTimeTaForLargeSquareTestInstances}
		\end{wrapfigure}
		
%		\addtolength{\tabcolsep}{-3pt}
		\begin{table}[htb!]
			\centering
			\scriptsize
			\caption{results for \emph{TSPLIB} test instances \\[0.2cm] \begin{tabular}{rll}$\bullet$ & $r^*(P)$ & objective function ratio \\ $\bullet$ & $t^*$ (A.~\ref{algorithm}) & running time of Algorithm~\ref{algorithm} in seconds \\ $\bullet$ & $t^*$ (B.) & running time of blossom algorithm in seconds\end{tabular}}
			\begin{tabular}{l | r S[round-mode=places, round-precision=4] S[round-mode=places, round-precision=2, table-number-alignment=right] S[round-mode=places, round-precision=2, table-number-alignment=right] | r S[round-mode=places, round-precision=4] S[round-mode=places, round-precision=2, table-number-alignment=right] S[round-mode=places, round-precision=2, table-number-alignment=right]}
				\toprule
				test instance & $n$ & {$r^*(P)$} & {$t^*$ (A.~\ref{algorithm})} & {$t^*$ (B.)} & $n$ & {$r^*(P)$} & {$t^*$ (A.~\ref{algorithm})} & {$t^*$ (B.)} \\
				\midrule
				eil51   & 50  & 0.9999 & 0.012981 & 0.110337 & 51  & 1.0000 & 0.005861 & 0.098140 \\
				berlin52& 52  & 0.9997 & 0.020907 & 0.129080 & 51  & 1.0000 & 0.009742 & 0.122869 \\
				st70    & 70  & 1.0000 & 0.025741 & 0.298391 & 69  & 1.0000 & 0.011388 & 0.241630 \\
				eil76   & 76  & 0.9994 & 0.046058 & 0.423416 & 75  & 1.0000 & 0.015038 & 0.353754 \\
				pr76    & 76  & 1.0000 & 0.066576 & 0.385726 & 75  & 0.9999 & 0.015116 & 0.305851 \\
				rat99   & 98  & 0.9999 & 0.057341 & 0.922280 & 99  & 0.9999 & 0.014964 & 0.964101 \\
				eil101  & 100 & 0.9999 & 0.047792 & 0.936372 & 101 & 1.0000 & 0.016517 & 1.012127 \\
				kroA100 & 100 & 0.9998 & 0.062809 & 0.972304 & 99  & 0.9998 & 0.011512 & 0.918672 \\
				kroB100 & 100 & 1.0000 & 0.034210 & 1.018549 & 99  & 1.0000 & 0.018028 & 0.897927 \\
				kroC100 & 100 & 0.9934 & 0.037160 & 0.959436 & 99  & 0.9988 & 0.025384 & 0.882519 \\
				kroD100 & 100 & 0.9986 & 0.017929 & 0.963303 & 99  & 1.0000 & 0.012598 & 0.958973 \\
				kroE100 & 100 & 0.9995 & 0.027297 & 1.019091 & 99  & 0.9999 & 0.017539 & 0.979591 \\
				rd100   & 100 & 1.0000 & 0.025779 & 0.918552 & 99  & 1.0000 & 0.015812 & 0.959041 \\
				lin105  & 104 & 0.9990 & 0.030464 & 1.112171 & 105 & 0.9985 & 0.024531 & 1.209141 \\
				pr107   & 106 & 1.0000 & 0.033504 & 1.047296 & 107 & 1.0000 & 0.022718 & 1.049516 \\
				pr124   & 124 & 0.9972 & 0.026086 & 1.685298 & 123 & 0.9987 & 0.031493 & 1.849608 \\
				bier127 & 126 & 0.9993 & 0.031775 & 1.745880 & 127 & 0.9996 & 0.021651 & 2.069829 \\
				ch130   & 130 & 0.9994 & 0.036601 & 1.681468 & 129 & 1.0000 & 0.033667 & 2.169848 \\
				pr136   & 136 & 1.0000 & 0.039932 & 1.308288 & 135 & 1.0000 & 0.028960 & 2.057633 \\
				pr144   & 144 & 1.0000 & 0.026346 & 2.384315 & 143 & 1.0000 & 0.020058 & 2.896932 \\
				ch150   & 150 & 1.0000 & 0.050290 & 2.550591 & 149 & 1.0000 & 0.037905 & 3.514448 \\
				kroA150 & 150 & 0.9999 & 0.087128 & 2.605742 & 149 & 0.9999 & 0.019012 & 3.386557 \\
				kroB150 & 150 & 0.9993 & 0.038484 & 2.675995 & 149 & 1.0000 & 0.037290 & 3.407082 \\
				pr152   & 152 & 0.9997 & 0.042563 & 2.787660 & 151 & 0.9999 & 0.025629 & 3.813904 \\
				u159    & 158 & 0.9985 & 0.089391 & 3.133735 & 159 & 0.9998 & 0.035191 & 4.227785 \\
				rat195  & 194 & 1.0000 & 0.119211 & 5.896169 & 195 & 1.0000 & 0.050618 & 7.640548 \\
				d198    & 198 & 0.9994 & 0.136808 & 6.682009 & 197 & 0.9994 & 0.051938 & 8.851830 \\
				kroA200 & 200 & 0.9984 & 0.164377 & 6.752514 & 199 & 0.9999 & 0.052375 & 8.625778 \\
				kroB200 & 200 & 0.9999 & 0.136509 & 7.710802 & 199 & 1.0000 & 0.050667 & 8.656638 \\
				ts225   & 224 & 0.9999 & 0.194858 & 7.742967 & 225 & 1.0000 & 0.048761 & 3.848079 \\
				tsp225  & 224 & 0.9992 & 0.152496 & 12.420489 & 225 & 0.9993 & 0.073348 & 11.253054 \\
				pr226   & 226 & 1.0000 & 0.182419 & 10.712864 & 225 & 1.0000 & 0.059282 & 10.538316 \\
				gil262  & 262 & 1.0000 & 0.377906 & 18.250593 & 261 & 1.0000 & 0.068532 & 19.531635 \\
				pr264   & 264 & 1.0000 & 0.226098 & 16.420738 & 263 & 1.0000 & 0.079944 & 19.996373 \\
				a280    & 280 & 1.0000 & 0.239382 & 22.249249 & 279 & 1.0000 & 0.054111 & 27.667645 \\
				pr299   & 298 & 0.9998 & 0.269150 & 34.435946 & 299 & 0.9998 & 0.078182 & 35.985746 \\
				lin318  & 318 & 0.9999 & 0.391829 & 43.860602 & 317 & 0.9999 & 0.084700 & 42.051415 \\
				rd400   & 400 & 1.0000 & 0.506220 & 93.802172 & 399 & 1.0000 & 0.108088 & 76.799473 \\
				fl417   & 416 & 0.9869 & 0.278456 & 97.906875 & 417 & 0.9862 & 0.068599 & 84.757042 \\
				pr439   & 438 & 0.9997 & 0.344209 & 137.435261 & 439 & 1.0000 & 0.126642 & 124.351815 \\
				pcb442  & 442 & 1.0000 & 0.163143 & 128.258327 & 441 & 1.0000 & 0.123453 & 124.737449 \\
				d493    & 492 & 0.9999 & 0.187882 & 200.981631 & 493 & 0.9999 & 0.142093 & 180.156362 \\
				rat575  & 574 & 1.0000 & 0.163309 & 306.056718 & 575 & 1.0000 & 0.131130 & 294.679415 \\
				u574    & 574 & 0.9988 & 0.401625 & 311.742138 & 573 & 1.0000 & 0.160370 & 285.200581 \\
				p654    & 654 & 0.9999 & 0.177513 & 420.417071 & 653 & 1.0000 & 0.126738 & 396.296182 \\
				d657    & 656 & 1.0000 & 0.266471 & 524.595582 & 657 & 1.0000 & 0.194798 & 431.970902 \\
				u724    & 724 & 1.0000 & 0.502331 & 686.767931 & 723 & 0.9999 & 0.143618 & 584.500201 \\
				rat783  & 782 & 1.0000 & 0.400313 & 897.050914 & 783 & 1.0000 & 0.228052 & 748.064758 \\
				dsj1000 & 1000& 0.9900 & 0.510419 & 1878.128062 & 999 & 0.9899 & 0.313561 & 1696.080558 \\
%				pr1002  & 1002& 0.9997 & 0.884317 & 1819.056926 & 1001& 0.9997 & 0.290929 & 1548.518796 \\
%				u1060   & 1060& 0.9993 & 1.329875 & 2002.159841 & 1059& 0.9995 & 0.299047 & 1718.125358 \\
%				vm1084  & 1084& 1.0000 & 0.510301 & 1518.364630 & 1083& 1.0000 & 0.308997 & 1301.293967 \\
				\bottomrule
			\end{tabular}
			\label{table:resultsForTSPLIBTestInstances}
		\end{table}
%		\addtolength{\tabcolsep}{3pt}
		
		At the same time, Algorithm~\ref{algorithm} requires only a fraction of the time needed by the blossom algorithm to obtain an optimal solution, as summed up in Tables~\ref{table:resultsForRandomTestInstances} and \ref{table:resultsForTSPLIBTestInstances} and visualised in Figures~\ref{figure:runningTimeTaForSquareAndCircleTestInstances} and \ref{figure:runningTimeTaForLargeSquareTestInstances}. To demonstrate the possibility of using our algorithm for larger test instances too, we created \emph{square} test instances (see Section~\ref{subsection:benchmarkInstancesAndTestEnvironment} for an exact specification) with $n = 1000, 2000, 3000, \ldots, 50000$ vertices (five of them for each size) and computed the mean running times our algorithm spent to solve them. The results, visualised in Figure~\ref{figure:runningTimeTaForLargeSquareTestInstances}, demonstrate the $O(n \log n)$ running time of our algorithm. It should also be pointed out that we implemented all our programs in Python. For practical use, e.g., in industry, compiled programming languages like Julia, C++, or Rust could make it possible to solve significantly larger instances as well.
		
\section{Final notes, conclusions and outlook}
	\label{section:finalNotesConclusionsAndOutlook}
	In this paper, we present a novel algorithm for solving the \emph{Euclidean maximum weight matching problem} (\emph{Euclidean MWM}). The computational results clearly demonstrate the effectiveness of the presented algorithm, which produces solutions of near-optimal or optimal quality that improves with instance size across all types of test instances used, while being substantially faster than all known optimal-solution algorithms.
	
	Regarding the solution quality, the observed trend provides strong evidence for the conjecture that Algorithm~\ref{algorithm} becomes even asymptotically optimal; \ie, that its solutions approach the optimum as the number of vertices $n$ tends to infinity. In fact, some theoretical results can be proved if the points are uniformly distributed in the Euclidean plane within a circle; a complete proof of asymptotic optimality in this case, together with a corresponding proof for the general case (\ie\ for instances that do not arise from uniformly distributed points), is currently under investigation and will be the subject of a subsequent paper.
			
	From the running-time perspective, Algorithm~\ref{algorithm} requires only a fraction of the time needed by the blossom algorithm to obtain an optimal solution. In fact, Algorithm~\ref{algorithm} guarantees an $O(n \log n)$ running time, whereas the blossom algorithm is $O(n^3)$ (see, \eg,~\cite{Edmonds:MaximumMatchingAndAPolyhedronWithO1Vertices,Gabow:AnEfficientImplementationOfEdmondsAlgorithmForMaximumMatchingOnGraphs,Lawler:CombinatorialOptimizationNetworksAndMatroids}). Although an optimal Euclidean MWM can be found in $O(n^{2.5})$ time using the algorithm of \textsc{Micali} and \textsc{Vazirani} (see~\cite{MicaliVazirani:An0SqrtVEAlgoithmForFindingMaximumMatchingInGeneralGraphs,Vazirani:MaximumMatchingAndAPolyhedronWith01Vertices}), the time complexity of $O(n \log n)$ can be reached only by heuristics yielding significantly worse solutions. Even the algorithm of \textsc{Duan} and \textsc{Pettie} (see~\cite{DuanPettie:LinearTimeApproximationForMaximumWeightMatching}), which guarantees a $(1 - \epsilon)$-approximation ratio and is therefore an FPTAS, is quadratic in the number of vertices $n$.
			
	Taken together, these results demonstrate that Algorithm~\ref{algorithm} offers an exceptional combination of solution quality and computational efficiency, making it a highly attractive approach for practical applications. In particular, its near-optimal performance and $O(n\log n)$ running time make it particularly well-suited for large and very large instances, for which exact methods quickly become computationally prohibitive.
	
\begin{credits}
\subsubsection{\ackname} We want to thank Klaus Ederer for his preliminary computational tests on this problem, which he carried out during his Bachelor's studies at the Technical University of Leoben.

\subsubsection{\discintname}
The authors have no competing interests to declare that are relevant to the content of this article.
\end{credits}
\bibliographystyle{splncs04}
\bibliography{EuclideanMaximumWeightMatchingProblem}

\begin{thebibliography}{10}
\providecommand{\url}[1]{\texttt{#1}}
\providecommand{\urlprefix}{URL }
\providecommand{\doi}[1]{https://doi.org/#1}

\bibitem{AichholzerFischerFischerMeierPferschyPilzStanek:MinimizationAndMaximizationVersionsOfTheQuadraticTravellingSalesmanProblem}
Aichholzer, O., Fischer, A., Fischer, F., Meier, F.J., Pferschy, U., Pilz, A.,
  Stan\v{e}k, R.: Minimization and maximization versions of the quadratic
  travelling salesman problem. Optimization  \textbf{66}(4),  521--546 (2017)

\bibitem{Avis:ASurveyOfHeuristicsForTheWeightedMatchingProblem}
Avis, D.: A survey of heuristics for the weighted matching problem. Networks
  \textbf{13}(4),  475--493 (1983)

\bibitem{BaumannGoldschmidtHochbaum:AFastAlgorithmForEuclideanMaximumWeightNonBipartiteMatching}
Baumann, P., Goldschmidt, O., Hochbaum, D.S.: A fast algorithm for euclidean
  maximum weight non-bipartite matching. In: Castrillon-Santana, M., Riccio,
  D., Fred, A., Marsico, M.D. (eds.) Proceedings of the 15th International
  Conference on Pattern Recognition Applications and Methods ({ICPRAM} 2026).
  vol.~55, pp. 411--418. SCITEPRESS -- Science and Technology Publications,
  Lda. (2026)

\bibitem{DrakeHougardy:ASimpleApproximationAlgorithmForTheWeightedMatchingProblem}
Drake, D.E., Hougardy, S.: A simple approximation algorithm for the weighted
  matching problem. Information Processing Letters  \textbf{85}(4),  211--213
  (2003)

\bibitem{DuanPettie:LinearTimeApproximationForMaximumWeightMatching}
Duan, R., Pettie, S.: Linear-time approximation for maximum weight matching.
  Journal of the ACM  \textbf{61}(1),  1--23 (2014)

\bibitem{Edmonds:MaximumMatchingAndAPolyhedronWithO1Vertices}
Edmonds, J.: Maximum matching and a polyhedron with 0,1-vertices. Journal of
  Research of the National Bureau of Standards---B. Mathematics and
  Mathematical Physics  \textbf{69B}(1 and 2),  125--130 (1965)

\bibitem{Gabow:AnEfficientImplementationOfEdmondsAlgorithmForMaximumMatchingOnGraphs}
Gabow, H.N.: An efficient implementation of edmonds' algorithm for maximum
  matching on graphs. Journal of the ACM  \textbf{23}(2),  221--234 (1976)

\bibitem{Galil:EfficientAlgorithmsForFindingMaximumMatchingInGraphs}
Galil, Z.: Efficient algorithms for finding maximum matching in graphs. ACM
  Computing Surveys  \textbf{18}(1),  23--38 (1986)

\bibitem{HagbergSchultSwart:ExploringNetworkStructureDynamicsAndFunctionUsingNetworkX}
Hagberg, A.A., Schult, D.A., Swart, P.J.: Exploring network structure,
  dynamics, and function using {NetworkX}. In: Varoquaux, G., Vaught, T.,
  Millman, J. (eds.) Proceedings of the 7th Python in Science Conference. pp.
  11--15. Pasadena, CA USA (2008)

\bibitem{Kolmogorov:BlossomVANewImplementationOfAMinimumCostPerfectMatchingAlgorithm}
Kolmogorov, V.: Blossom v: a new implementation of a minimum cost perfect
  matching algorithm. Mathematical Programming Computation  \textbf{1},  43--67
  (2009)

\bibitem{KorteVygen:CombinatorialOptimizationTheoryAndAlgorithms}
Korte, B., Vygen, J.: Combinatorial Optimization: Theory and Algorithms.
  Springer Berlin, Heidelberg, sixth edition edn. (2018)

\bibitem{Lawler:CombinatorialOptimizationNetworksAndMatroids}
Lawler, E.L.: Combinatorial Optimization: Networks and Matroids. Holt, Rinehart
  and Winston (1976)

\bibitem{LovaszPlummer:MatchingTheory}
Lovász, L., Plummer, M.D.: Matching Theory. American Mathematical Society
  (2009)

\bibitem{LuGreevyXuBeck:OptimalNonbipartiteMatchingAndItsStatisticalApplications}
Lu, B., Greevy, R., Xu, X., Beck, C.: Optimal nonbipartite matching and its
  statistical applications. The American Statistician  \textbf{65}(1),  21--30
  (2011)

\bibitem{MicaliVazirani:An0SqrtVEAlgoithmForFindingMaximumMatchingInGeneralGraphs}
Micali, S., Vazirani, V.V.: An {$O(\sqrt{|V|} E)$} algoithm for finding maximum
  matching in general graphs. In: Proceedings of the 21st Annual Symposium on
  Foundations of Computer Science ({SFCS '80}). pp. 17--27. {IEEE} Computer
  Society (1980)

\bibitem{NainiUnnikrishnanThiranVetterli:WhereYouAreIsWhoYouAreUserIdentificationByMatchingStatistics}
Naini, F.M., Unnikrishnan, J., Thiran, P., Vetterli, M.: Where you are is who
  you are: User identification by matching statistics. IEEE Transactions on
  Information Forensics and Security  \textbf{11}(2),  358--372 (2016)

\bibitem{Preis:LinearTime12ApproximationAlgorithmForMaximumWeightedMatchingInGeneralGraphs}
Preis, R.: Linear time $\frac{1}{2}$-approximation algorithm for maximum
  weighted matching in general graphs. In: Meinel, C., Tison, S. (eds.)
  Proceedings of the 16th Annual Symposium on Theoretical Aspects of Computer
  Science ({STACS 99}). Lecture Notes in Computer Science, vol.~1563, pp.
  259--269. Springer Berlin, Heidelberg (1999)

\bibitem{Reinelt:TSPLIB}
Reinelt, G.: {TSPLIB}. Website (1995), available at
  \url{http://comopt.ifi.uni-heidelberg.de/software/TSPLIB95/}

\bibitem{SchenkGartner:OnFastFactorizationPivotingMethodsForSparseSymmetricIndefiniteSystems}
Schenk, O., Gärtner, K.: On fast factorization pivoting methods for sparse
  symmetric indefinite systems. Electronic Transactions on Numerical Analysis
  \textbf{23},  158--179 (2006)

\bibitem{Schrijver:CombinatorialOptimizationPolyhedraAndEfficiency}
Schrijver, A.: Combinatorial Optimization: Polyhedra and Efficiency.
  Springer-Verlag Berlin Heidelberg (2003)

\bibitem{Vazirani:MaximumMatchingAndAPolyhedronWith01Vertices}
Vazirani, V.V.: Maximum matching and a polyhedron with 0,1-vertices.
  Mathematics of Operations Research  \textbf{49}(3),  2009--2047 (2024)

\bibitem{WuLi:SolvingMaximumWeightedMatchingOnLargeGraphsWithDeepReinforcementLearning}
Wu, B., Li, L.: Solving maximum weighted matching on large graphs with deep
  reinforcement learning. Information Sciences  \textbf{614},  400--415 (2022)

\bibitem{WuZhong:FusionBlossomFastMWPMDecodersForQEC}
Wu, Y., Zhong, L.: Fusion blossom: Fast {MWPM} decoders for {QEC}. In:
  Proceedings of the 2023 IEEE International Conference on Quantum Computing
  and Engineering ({QCE}). vol.~01, pp. 928--938 (2023)

\bibitem{XuChu:AMatchingBasedDecomposerForDoublePatterningLithography}
Xu, Y., Chu, C.: A matching based decomposer for double patterning lithography.
  In: Proceedings of the 19th International Symposium on Physical Design ({ISPD
  '10}). pp. 121--126. Association for Computing Machinery (2010)

\end{thebibliography}
%Author, F.: Article title. Journal \textbf{2}(5), 99--110 (2016)
%Author, F., Author, S.: Title of a proceedings paper. In: Editor,
%F., Editor, S. (eds.) CONFERENCE 2016, LNCS, vol. 9999, pp. 1--13.
%Springer, Heidelberg (2016). \doi{10.10007/1234567890}
%Author, F., Author, S., Author, T.: Book title. 2nd edn. Publisher,
%Author, A.-B.: Contribution title. In: 9th International Proceedings
%on Proceedings, pp. 1--2. Publisher, Location (2010)
%LNCS Homepage, \url{http://www.springer.com/lncs}, last accessed 2023/10/25
%\end{thebibliography}
\end{document}